\documentclass[11pt,fleqn,a4paper]{article}

\usepackage{amsmath,amssymb,amsthm}
\usepackage[mathscr]{eucal}
\usepackage{graphicx}
\usepackage[labelsep=period]{caption}

\usepackage{hyperref}
\hypersetup{colorlinks, linkcolor=blue, citecolor=blue, urlcolor=blue}

\newcommand{\p}{\partial}
\newcommand{\const}{{\rm const}}
\newcommand{\spanindex}{{\mbox{\tiny$\langle\,\rangle$}}}

\newtheorem{theorem}{Theorem}
\newtheorem{lemma}[theorem]{Lemma}
\newtheorem{corollary}[theorem]{Corollary}
\newtheorem{proposition}[theorem]{Proposition}
\newtheorem*{proposition*}{Proposition}
{\theoremstyle{definition}
\newtheorem{definition}[theorem]{Definition}

\newtheorem{remark}[theorem]{Remark}

}

\newcommand{\todo}[1][\null]{\ensuremath{\clubsuit}}

\newcommand{\noprint}[1]{}

\begin{document}
\par\noindent {\LARGE\bf
Realizations of Lie algebras on the line\\
and the new group classification of (1+1)-dimensional\\ generalized nonlinear Klein--Gordon equations
\par}

{\vspace{4mm}\par\noindent {\large Vyacheslav M.\ Boyko$^\dag$, Oleksandra V.\ Lokaziuk$^{\dag}$ and Roman O.\ Popovych$^{\dag\ddag}$
} \par\vspace{2mm}\par}

{\vspace{2mm}\par\noindent {\it
$^{\dag}$Institute of Mathematics of NAS of Ukraine, 3 Tereshchenkivska Str., 01024 Kyiv, Ukraine
}}
{\vspace{2mm}\par\noindent {\it
$^\ddag$Fakult\"at f\"ur Mathematik, Universit\"at Wien, Oskar-Morgenstern-Platz 1, 1090 Wien, Austria
}}

{\vspace{2mm}\par\noindent {\it
\textup{E-mail:} boyko@imath.kiev.ua, sasha.lokazuik@gmail.com, rop@imath.kiev.ua
}\par}

\vspace{8mm}\par\noindent\hspace*{10mm}\parbox{140mm}{\small
Essentially generalizing Lie's results, 
we prove that the contact equivalence groupoid 
of a class of (1+1)-dimensional generalized nonlinear Klein--Gordon equations 
is the first-order prolongation of its point equivalence groupoid, 
and then we carry out the complete group classification of this class.  
Since it is normalized, the algebraic method of group classification
is naturally applied here.
Using the specific structure of the equivalence group of the class,
we essentially employ
the classical Lie theorem on realizations of Lie algebras by vector fields on the line.
This approach allows us to enhance previous results on Lie symmetries of equations from the class
and substantially simplify the proof.
After finding a number of integer characteristics of cases of Lie-symmetry extensions
that are invariant under action of the equivalence group of the class under study,
we exhaustively describe successive Lie-symmetry extensions within this class.
}\par\vspace{4mm}

\noprint{


}

\section{Introduction}\label{sec:Introduction}

Quasilinear second-order hyperbolic equations model many phenomena and processes in physics and mathematics, 
especially, various kinds of wave propagation, 
see \cite{dodd1982A,fox2011a,fush1993A,poly2012A} and references therein.
Such equations even with two independent variables are important 
in a number of areas, including differential geometry, quantum field theory, cosmology, hydro- and gas dynamics, 
superconductivity, crystal dislocation, waves in ferromagnetic materials, nonlinear optics, low temperature physics, to name a few. 
This is why these equations have been and are intensively studied in many branches of mathematics, in particular, 
within the framework of integrability theory and symmetry analysis of differential equations.  

In the present paper, we carry out the exhaustive group classification of the class
of (1+1)-dimensional generalized nonlinear Klein--Gordon equations of the form
\begin{gather}\label{eq:NonlinKGEqs}
u_{tx}=f(t,x,u) \quad\mbox{with}\quad f_{uu}\ne0
\end{gather}
in the light-cone (characteristic) coordinates, 
which we denote by~$\mathcal K$ and simultaneously refer as the class~\eqref{eq:NonlinKGEqs} and the class~$\mathcal K$.
Here $u=u(t,x)$ is the unknown function of the independent variables~$(t,x)$,
and subscripts of functions denote derivatives with respect to the corresponding variables,
e.g., $u_{tx}:=\p^2 u/\p t\,\p x$ and $f_{uu}:=\p^2 f/\p u^2$.
The arbitrary element~$f$ of the class~$\mathcal K$ runs through the set of smooth functions of $(t,x,u)$
that are not affine in~$u$.
The last constraint is imposed on~$f$ for excluding linear equations from the class~$\mathcal K$, 
which is natural in view of several arguments, see Remark~\ref{rem:NonlinKGEqsOnContactTrans}.
The problem of studying the class~$\mathcal K$ within the framework of group analysis of differential equations
was posed by Sophus Lie~\cite{lie1881b}. 

We discover four important properties of the class~$\mathcal K$, which allows us to obtain stronger results than the ordinary group classification of this class 
and to simplify all computations. We prove that, firstly, the class~$\mathcal K$ is normalized with respect to its point equivalence group~$G^\sim$ 
and, secondly, the first prolongations of its point admissible transformations 
exhaust its contact admissible transformations. 
In view of these two properties, 
the classification of Lie symmetries of equations from the class~$\mathcal K$ up to the $G^\sim$-equivalence, 
roughly speaking, coincides with the similar classification up to the general point equivalence 
and with the classifications of continuous contact symmetries of these equations 
modulo the equivalences generated by the contact equivalence group and groupoid of~$\mathcal K$. 
Moreover, this classification problem can be effectively solved by the algebraic method. 
Thirdly, the dimensions of the maximal Lie invariance algebras of equations from the class~$\mathcal K$ 
are not greater than four, except the equations that are $G^\sim$-equivalent to the Liouville equation 
and whose maximal Lie invariance algebras are infinite-dimensional. 
The fourth property is that each point transformation between any two equations from the class~$\mathcal K$ 
is projectable both on the space with coordinate~$t$ and on the space with coordinate~$x$.
As a result, all the point-transformation structures associated with the class~$\mathcal K$, 
including its equivalence group, its equivalence algebra 
and the maximal Lie invariance algebras of equations from this class, are also projectable on the same two spaces. 
This twofold projectability leads to the possibility of extended application of  
the following classical Lie theorem on realizations of finite-dimensional Lie algebras by vector fields on the line \cite[Satz~6, p.~455]{lie1880a}
in the course of the group classification, 
see also \cite[Theorem~2.70]{olve1995A}, \cite{popo2003a} and references therein. 

\begin{theorem}[Lie theorem]
Inequivalent (up to the local diffeomorphisms of the line)
realizations of finite-dimensional Lie algebras by vector fields on the $t$-line are exhausted
by the algebras
\[
\{0\},\quad \langle\p_t\rangle,\quad \langle\p_t,\,t\p_t\rangle,\quad\langle\p_t,\,t\p_t,\,t^2\p_t\rangle.
\]
\end{theorem}

Previously, the Lie theorem was applied to group classification 
of classes of evolution and Schr\"odinger equations \cite{bihlo2017a,kuru2018a,kuru2020a,opan2017a}
and the class of linear ordinary differential equations of an arbitrary fixed order $r\geqslant2$ \cite[Section~3]{boyk2015a}, 
where there is the similar projectability to the single space coordinatized by the time variable or the independent variable, respectively.

The most important and well-studied subclass~$\mathcal K_9$ of the class~$\mathcal K$ 
is singled out from~$\mathcal K$ by the additional auxiliary equations $f_t=f_x=0$, 
i.e., this subclass consists of the nonlinear Klein--Gordon equations, 
which are of the form $u_{tx}=f(u)$ with $f_{uu}\ne0$; 
cf.\ Remark~\ref{rem:NonlinKGEqsOnMaxLieSymExtensions} for justifying the notation.
The subclass~$\mathcal K_9$ contains a number of famous equations, 
which we present in the canonical forms, 
where the constant parameters are removed by equivalence transformations:
\begin{itemize}\itemsep=-.5ex
\item the Liouville equation $u_{tx}={\rm e}^u$,
\item the Tzitzeica equation $u_{tx}={\rm e}^u\pm{\rm e}^{-2u}$ called also the Dodd--Bullough--Mikhailov equation, 
\item the sine-Gordon (or Bonnet) equation $u_{tx}=\sin u$,  
\item the sinh-Gordon equation $u_{tx}=\sinh u$, 
\item the double sine-Gordon equation $u_{tx}=\sin u+C\sin 2u$ with $C\ne0$, 
\end{itemize}
see Section~7.5.1 in~\cite{poly2012A} and references therein. 
Contact symmetry transformations of equations from the subclass~$\mathcal K_9$ 
were described by Sophus Lie himself~\cite{lie1881b}. 
In particular, any such transformation was proved to be the first prolongation of a point transformation. 
Then Lie singled out the Liouville equations, where $f(u)=C{\rm e}^{\kappa u}$ with nonzero constants~$\kappa$ and~$C$, 
as the only equations in~$\mathcal K_9$ 
admitting infinite-dimensional point symmetry groups. 
The other equations were shown to possess only point symmetry transformations, 
where each of the $t$-, $x$- and $u$-components depends only on the respective variable, and this dependence is affine.  
In the introduction of~\cite{lie1881b}, Lie remarked that results obtained therein can be extended to the class~$\mathcal K$.
In fact, the present paper essentially generalizes several extensions of Lie's results to the class~$\mathcal K$. 

Specific equations from the subclass~$\mathcal K_9$ and the subclass~$\mathcal K_9$ itself 
were intensively studied within the framework of symmetry analysis of differential equations. 
In particular, generalized symmetries of equations from~$\mathcal K_9$ 
with characteristics not depending on $(t,x)$ were classified over the complex field in~\cite{zhib1979a}. 
The Liouville equation, the sine-Gordon equation and the Tzitzeica equation 
were singled out in the subclass~$\mathcal K_9$
as the only equations with infinite-dimensional algebras of such symmetries, 
see also \cite[Section~21.2]{ibra1985A}. 
The same equations were also singled out 
in~\cite{dodd1977a} as the only equations in~$\mathcal K_9$ 
admitting infinite-dimensional spaces of conservation laws with so-called ``polynomial densities''
and in~\cite{clar1986a,mcle1983a} as the only ones with the Painlev\'e property among the equations $u_{tx}=f(u)$, 
where the function~$f$ is a linear combination of exponential functions ${\rm e}^{j\alpha u}$ 
for some fixed nonzero complex constant~$\alpha$ and $j\in\mathbb Z$. 
Equations from the subclass~$\mathcal K_9$ admitting 
nonlinear separation of variables in the standard spacetime coordinates 
to two first-order ordinary differential equations were classified in~\cite{zhda1994a}; 
see also references therein and~\cite{grun1992a} for other kinds of nonlinear separation of variables for these equations. 
The classification of local conservation laws of equations from the subclass~$\mathcal K_9$ 
over the complex field was begun in~\cite{fox2011a}. 
Singular reduction operators~\cite{boyk2016a,kunz2008a}, i.e., singular nonclassical (or conditional, or \mbox{$Q$-conditional}) symmetries, 
of all the equations of the form $u_{tx}=f(u)$ were exhaustively studied in~\cite{kunz2008a}. 
At the same time, there are still no complete classifications of generalized symmetries, local conservation laws 
and regular reduction operators of equations from the subclass~$\mathcal K_9$ 
as well as no exhaustive classifications of such equations admitting nonlinear separation of variables or the Painlev\'e property, 
not to mention the entire class~$\mathcal K$.
An exception is the general description of regular reduction operators for equations from the class~$\mathcal K$
that was given in~\cite{yeho2010a}. 
\looseness=-1

The framework of the algebraic method of group classification originated in 
Lie's classification of second-order ordinary differential equations~\cite{lie1891A} 
but it became a common tool of group analysis of differential equations considerably later, only since the 1990s, 
although its applications to problems of complete group classification still involved the normalization property implicitly 
\cite{basa2001a,gagn1993a,gaze1992a,gung2004a,lahn2005a,lahn2006a,maga1993a,zhda1999d}.
When straightforwardly applied to non-normalized classes of differential equations, 
the algebraic method results in the so-called preliminary group classification of such classes \cite{akha1991a,bihl2012b,card2011a,ibra1991a}.
The algebraic method is usually used to solve group classification problems for classes of differential equations 
with arbitrary elements depending on several arguments, 
for which the direct method of group classification, 
including the method of furcate splitting \cite{bihl2020a,niki2001a,opan2020b} as its most advanced version, is unproductive.
To carry out the group classification of the class~$\mathcal K$, 
we use the advanced version of the algebraic method, which is based on 
the normalization of the class of differential equations to be classified~\cite{popo2006b,popo2005c,popo2010a}
and involves the classification of appropriate subalgebras~\cite{bihl2012b,card2011a} of the corresponding equivalence algebra. 
This version of the algebraic method was suggested in~\cite{popo2005c,popo2010a} 
and was effectively applied to solving group classification problems for various classes of differential equations 
\cite{bihl2012b,bihlo2017a,boyk2015a,kuru2018a,kuru2020a,opan2017a,popo2012a,popo2005c,popo2010a,vane2020b}. 

The papers \cite{bihl2012b,lahn2005a,lahn2006a,vane2020b} are especially relevant 
in the context of the present paper since they are devoted to group classification 
of classes of quasilinear hyperbolic second-order equations with two independent variables using the algebraic method.
See also references therein for group classifications of other classes of such equations. 
In particular, the group classification problem for the superclass~$\breve{\mathcal K}$ of~$\mathcal K$ 
that is constituted by the equations of the form $u_{\check t\check t}-u_{\check x\check x}=f(\check t,\check x,u,u_{\check x})$ 
in the standard spacetime coordinates $(\check t,\check x)=(x+t,x-t)$
was studied in the seminal papers~\cite{lahn2005a,lahn2006a}. 
The superclass~$\breve{\mathcal K}$ was partitioned into four subclasses, 
which are in fact normalized and are not related by point transformations to each other, 
and~$\mathcal K$ is one of these subclasses (under using the light-cone coordinates).
As a result, the group classification problem for the entire superclass~$\breve{\mathcal K}$ 
was split into four group classification problems for the subclasses 
that each was separately studied within the framework of the algebraic method of group classification. 
Unfortunately, the consideration of the class~$\mathcal K$ had several drawbacks 
(see the second paragraph of Section~\ref{sec:Conclusion} below), 
and hence Lie symmetries of equations from this class have needed a more accurate and comprehensive classification, 
which is done in the present paper. 
The group classification problems for the non-normalized class of quasilinear hyperbolic and elliptic equations 
of the form $u_{\check t\check t}-h(\check x,u,u_{\check x})u_{\check x\check x}=f(\check x,u,u_{\check x})$
up to the equivalences generated by the corresponding equivalence group and groupoid, respectively, 
were exhaustively solved in~\cite{vane2020b} 
using an original version of the algebraic method of group classification 
for non-normalized classes of differential equations. 
When written in the light-cone coordinates, this class nontrivially intersects the class~$\mathcal K$. 
 
The further organization of the present paper is as follows.
In Section~\ref{sec:NonlinKGEqsPreliminaryAnalysis},
we prove that any contact-transformation structure associated with equations from the class~\eqref{eq:NonlinKGEqs} 
is the first-order prolongation of its point-transformation counterpart, 
thus justifying the restriction of the further consideration to point-transformation structures. 
Then we show the normalization of the class~\eqref{eq:NonlinKGEqs} (with respect to point transformations)
and construct its (point) equivalence group~$G^\sim$ and its (point) equivalence algebra~$\mathfrak g^\sim$.
Therein, we also single out the equations in the class~\eqref{eq:NonlinKGEqs}
with infinite-dimensional Lie invariance algebras 
and derive properties of finite-dimensional appropriate subalgebras of the projection~$\varpi_*\mathfrak g^\sim$ of~$\mathfrak g^\sim$ onto
the space with coordinates~$(t,x,u)$, which are important in the course of the group classification of this class.
The main result of the group classification,
which is a complete list of $G^\sim$-inequivalent Lie-symmetry extensions within the class~\eqref{eq:NonlinKGEqs},
is given by Theorem~\ref{thm:NonlinKGEqsGroupClassification} in Section~\ref{sec:NonlinKGEqsResults}.
The structure of the partially ordered set of such extensions is represented as a Hasse diagram in Figure~\ref{fig:HasseDiagram}. 
We also discuss relations between Lie-symmetry extensions via limit processes. 
Section~\ref{sec:NonlinKGEqsProof} is devoted to the proof of Theorem~\ref{thm:NonlinKGEqsGroupClassification}
and its analysis.
We find a number of $G^\sim$-invariant integer characteristics of subalgebras of~$\mathfrak g^\sim$ or, equivalently,
of~$\varpi_*\mathfrak g^\sim$, which allow us to completely identify $G^\sim$-inequivalent cases of Lie-symmetry extensions
within the class~\eqref{eq:NonlinKGEqs}.
In Section~\ref{sec:SuccessiveLie-SymExtensions},
we use these characteristics to distinguish, modulo the $G^\sim$-equivalence,
successive Lie-symmetry extensions among the found ones,
thus exhaustively describing the structure of partially ordered set of $G^\sim$-inequivalent Lie-symmetry extensions within the class~\eqref{eq:NonlinKGEqs}.
Possible ways for the group classifications of subclasses of the class~\eqref{eq:NonlinKGEqs}
are analyzed in Section~\ref{sec:OnGroupClassificationOfSubclasses}.
As an example, we use results of~\cite{vane2020b}
to carry out the group classification of the important subclass~$\mathcal K_2$ associated with the constraint $f_x+f_t=0$
up to the equivalence generated by the equivalence group~$G^\sim_2$ of this subclass.
In Section~\ref{sec:Conclusion}, we discuss the obtained results and overview related problems for the further study.

\section{Preliminary analysis}
\label{sec:NonlinKGEqsPreliminaryAnalysis}

Consider the superclass~$\mathcal K_{\rm gen}$ of all the equation of the general form $u_{tx}=f(t,x,u)$, 
$\mathcal K\subset\mathcal K_{\rm gen}$.
For a fixed value of the arbitrary element~$f$,
let $K_f$ denote the equation from the class~$\mathcal K_{\rm gen}$ with this value of~$f$.
We begin with the study of contact admissible transformations within the subclass~$\bar{\mathcal K}$ of~$\mathcal K_{\rm gen}$
singled out by the constraint $f_u\ne0$, 
i.e., we attach to~$\mathcal K$ the linear equations of the form $u_{tx}=f(t,x,u)$ with $f_{uu}=0$ and $f_u\ne0$. 
We essentially generalize Lie's consideration in~\cite{lie1881b}.

\begin{lemma}\label{lem:GenNonlinKGEqsContactAdmTrans}
Any contact admissible transformation within the class~$\bar{\mathcal K}$
is the first-order prolongation of a point admissible transformation within this class.
\end{lemma}

\begin{proof}
We give a simple proof by the direct method. 
We fix a contact admissible transformation $\mathcal T=(f,\Phi,\tilde f)$ of the class~$\bar{\mathcal K}$.
Here $\Phi$ is a contact transformation with the independent variables $(t,x)$ and the dependent variable~$u$,
$\Phi$: $(\tilde t,\tilde x,\tilde u,\tilde u_{\tilde t},\tilde u_{\tilde x})=(T,X,U,U^t,U^x)$,
that maps the equation $K_f$: $u_{tx}=f(t,x,u)$
to the equation \smash{$K_{\tilde f}$}: $\tilde u_{\tilde t\tilde x}=\tilde f(\tilde t,\tilde x,\tilde u)$.
The functions $T$, $X$, $U$, $U^t$ and~$U^x$ defining the components of the transformation~$\Phi$ 
are smooth functions of~$(t,x,u,u_t,u_x)$ with $\big|\p(T,X,U,U^t,U^x)/\p(t,x,u,u_t,u_x)\big|\ne0$
that satisfy the contact condition
\begin{gather}\label{eq:GenNonlinKGEqsContactCondition}
\begin{split}
&U^t\mathrm D_tT+U^x\mathrm D_tX=\mathrm D_tU,\\
&U^t\mathrm D_xT+U^x\mathrm D_xX=\mathrm D_xU,
\end{split}
\end{gather}
where
$\mathrm D_t=\p_t+u_t\p_u+u_{tt}\p_{u_t}+u_{tx}\p_{u_x}+\cdots$ and
$\mathrm D_x=\p_x+u_x\p_u+u_{tx}\p_{u_t}+u_{xx}\p_{u_x}+\cdots$ 
are the operators of total derivatives with respect to~$t$ and~$x$, respectively.
Collecting coefficients of the second derivatives of~$u$ in the contact condition~\eqref{eq:GenNonlinKGEqsContactCondition}
leads to the system
\begin{gather}\label{eq:GenNonlinKGEqsSplitContactCondition}
\begin{split}
&U^tT_{u_t}+U^xX_{u_t}=U_{u_t},\quad U^t\hat{\mathrm D}_tT+U^x\hat{\mathrm D}_tX=\hat{\mathrm D}_tU,\\
&U^tT_{u_x}+U^xX_{u_x}=U_{u_x},\quad U^t\hat{\mathrm D}_xT+U^x\hat{\mathrm D}_xX=\hat{\mathrm D}_xU,
\end{split}
\end{gather}
where $\hat{\mathrm D}_t=\p_t+u_t\p_u$ and $\hat{\mathrm D}_x=\p_x+u_x\p_u$ are the truncated operators of total derivatives with respect to~$t$ and~$x$.
The condition that the transformation~$\Phi$ maps the equation $K_f$ to the equation $K_{\tilde f}$
is expanded via substituting the expression for~$\tilde u_{\tilde t\tilde x}$
in terms of the variables without tildes,
\begin{gather}\label{eq:GenNonlinKGEqsConditionForAdmTrans}\arraycolsep=.5ex
\left|\begin{array}{cc}\mathrm D_tU^x&\mathrm D_xU^x\\\mathrm D_tX  &\mathrm D_xX  \end{array}\right|=
\left|\begin{array}{cc}\mathrm D_tT  &\mathrm D_xT  \\\mathrm D_tU^t&\mathrm D_xU^t\end{array}\right|=(\Phi^*\tilde f)
\left|\begin{array}{cc}\mathrm D_tT  &\mathrm D_xT  \\\mathrm D_tX  &\mathrm D_xX  \end{array}\right|
\quad\mbox{on solutions of}\ K_f.
\end{gather}
Here $\Phi^*$ denotes the pullback by~$\Phi$, $\Phi^*\tilde f:=\tilde f(T,X,U)$.
The first equality in~\eqref{eq:GenNonlinKGEqsConditionForAdmTrans}
is a differential consequence of the system~\eqref{eq:GenNonlinKGEqsContactCondition}.
Splitting of the equation~\eqref{eq:GenNonlinKGEqsConditionForAdmTrans} with respect to~$u_{tt}$ and~$u_{xx}$
gives, in particular, the equations
\begin{subequations}\arraycolsep=.5ex
\begin{gather}\label{eq:GenNonlinKGEqsConditionForAdmTransA}
\left|\begin{array}{cc} U^x_{u_t}& U^x_{u_x}\\ X  _{u_t}& X  _{u_x}\end{array}\right|=
\left|\begin{array}{cc} T  _{u_t}& T  _{u_x}\\ U^t_{u_t}& U^t_{u_x}\end{array}\right|=(\Phi^*\tilde f)
\left|\begin{array}{cc} T  _{u_t}& T  _{u_x}\\ X  _{u_t}& X  _{u_x}\end{array}\right|,
\\[1ex]\label{eq:GenNonlinKGEqsConditionForAdmTransB}
\left|\begin{array}{cc} U^x_{u_t}&\hat{\mathrm D}_x U^x\\ X  _{u_t}&\hat{\mathrm D}_x X  \end{array}\right|=
\left|\begin{array}{cc} T  _{u_t}&\hat{\mathrm D}_x T  \\ U^t_{u_t}&\hat{\mathrm D}_x U^t\end{array}\right|=(\Phi^*\tilde f)
\left|\begin{array}{cc} T  _{u_t}&\hat{\mathrm D}_x T  \\ X  _{u_t}&\hat{\mathrm D}_x X  \end{array}\right|,
\\[1ex]\label{eq:GenNonlinKGEqsConditionForAdmTransC}
\left|\begin{array}{cc}\hat{\mathrm D}_t U^x& U^x_{u_x}\\\hat{\mathrm D}_t X  & X  _{u_x}\end{array}\right|=
\left|\begin{array}{cc}\hat{\mathrm D}_t T  & T  _{u_x}\\\hat{\mathrm D}_t U^t& U^t_{u_x}\end{array}\right|=(\Phi^*\tilde f)
\left|\begin{array}{cc}\hat{\mathrm D}_t T  & T  _{u_x}\\\hat{\mathrm D}_t X  & X  _{u_x}\end{array}\right|.
\end{gather}
\end{subequations}

Suppose that at least one of the derivatives~$T_{u_t}$, $T_{u_x}$, $X_{u_t}$ and $X_{u_x}$ does not vanish.
Up to the permutations of~$t$ and~$x$ and of~$\tilde t$ and~$\tilde x$, which are equivalence transformations of the class~$\bar{\mathcal K}$,
we can assume that $T_{u_t}\ne0$.
We denote
\[
\Lambda:=\frac{U^t_{u_t}-(\Phi^*\tilde f)X_{u_t}}{T_{u_t}},
\]
and thus $\Lambda$ is a function of~$(t,x,u,u_t,u_x)$.
This notation and the second equalities 
in the equations~\eqref{eq:GenNonlinKGEqsConditionForAdmTransA} and~\eqref{eq:GenNonlinKGEqsConditionForAdmTransB} imply the equations
\begin{gather*}
U^t_{u_t}=\Lambda T_{u_t}+(\Phi^*\tilde f) X_{u_t},\\
U^t_{u_x}=\Lambda T_{u_x}+(\Phi^*\tilde f) X_{u_x},\\
\hat{\mathrm D}_x U^t=\Lambda\hat{\mathrm D}_x T+(\Phi^*\tilde f)\hat{\mathrm D}_x X,
\end{gather*}
which are combined to the single equation
$\mathrm D_xU^t=\Lambda\mathrm D_xT+(\Phi^*\tilde f)\mathrm D_xX$
defined on the entire second-order jet space~$\mathrm J^2(\mathbb R^2_{t,x}\times\mathbb R_u)$ with the independent variables $(t,x)$ and the dependent variable~$u$.
Subtracting the last equation from the equality 
$\mathrm D_xU^t
=(\mathop{\rm pr}\nolimits_{(2)}\Phi)^*(\tilde u_{\tilde t\tilde t})\mathrm D_xT
+(\mathop{\rm pr}\nolimits_{(2)}\Phi)^*(\tilde u_{\tilde t\tilde x})\mathrm D_xX$, we derive the equation 
\begin{gather}\label{eq:GenNonlinKGEqsDerivedConditionForAdmTrans1}
\big((\mathop{\rm pr}\nolimits_{(2)}\Phi)^*(\tilde u_{\tilde t\tilde t})-\Lambda\big)\mathrm D_xT
+(\mathop{\rm pr}\nolimits_{(2)}\Phi)^*(\tilde u_{\tilde t\tilde x}-\tilde f)\mathrm D_xX=0
\end{gather}
on~$\mathrm J^2(\mathbb R^2_{t,x}\times\mathbb R_u)$.
Here $\mathop{\rm pr}\nolimits_{(2)}\Phi$ denotes the second-order prolongation of the contact transformation~$\Phi$. 
Restricting the equation~\eqref{eq:GenNonlinKGEqsDerivedConditionForAdmTrans1} on the manifold defined by~$K_f$ in~$\mathrm J^2(\mathbb R^2_{t,x}\times\mathbb R_u)$, 
where $u_{tx}=f$ and $(\mathop{\rm pr}\nolimits_{(2)}\Phi)^*(\tilde u_{\tilde t\tilde x}-\tilde f)=0$, leads to the equality 
\[
\big((\mathop{\rm pr}\nolimits_{(2)}\Phi)^*(\tilde u_{\tilde t\tilde t})-\Lambda\big)(\hat{\mathrm D}_xT+fT_{u_t}+T_{u_x}u_{xx})=0.
\]
Since the derivative~$\tilde u_{\tilde t\tilde t}$ is not constrained on the solutions of~\smash{$K_{\tilde f}$}, 
this implies the equation $\hat{\mathrm D}_xT+fT_{u_t}+T_{u_x}u_{xx}=0$, 
which can be split with respect to~$u_{xx}$ into the equations 
$T_{u_x}=0$ and $\hat{\mathrm D}_xT+fT_{u_t}=0$. 
In view of the first of these two equations, the second equation can further be split with respect to~$u_x$ 
into the equations $T_u=0$ and $T_x+fT_{u_t}=0$. 
Since $T_u=0$ and $f_u\ne0$, the equation $T_x+fT_{u_t}=0$ splits into $T_x=T_{u_t}=0$, 
which contradicts the inequality $T_{u_t}\ne0$.

Therefore, $T_{u_t}=T_{u_x}=X_{u_t}=X_{u_x}=0$.
Then the system~\eqref{eq:GenNonlinKGEqsSplitContactCondition} directly implies \mbox{$U_{u_t}=U_{u_x}=0$}.
Since the $t$-, $x$- and $u$-components of the contact transformation~$\Phi$
do not depend on the first-order derivatives~$u_t$ and~$u_x$, 
this transformation is a first-order prolongation of the point transformations with the same $t$-, $x$- and $u$-components.
\end{proof}

A generalization of Lemma~\ref{lem:GenNonlinKGEqsContactAdmTrans} for the class of equations of the form $u_{tx}=f(t,x,u,u_t,u_x)$ 
was proved in~\cite{popo2021a}.
In view of Lemma~\ref{lem:GenNonlinKGEqsContactAdmTrans},
all structures related to contact transformations of equations from the superclass~$\bar{\mathcal K}$
and all its subclasses, including the class~$\mathcal K$,
are the first-order prolongations of analogous structures related to point transformations of equations from the same classes.
These structures include equivalence groupoids, equivalence groups of these classes and the symmetry groups of equations from them. 
This is why we restrict the further consideration to point transformations within the class~$\mathcal K$.

The following lemma is an obvious corollary of Lemma~\ref{lem:GenNonlinKGEqsContactAdmTrans} and \cite[Theorem~4.3c]{king1998a}.

\begin{lemma}\label{lem:NonlinKGEqsGsim}
The class~\eqref{eq:NonlinKGEqs} is normalized in the usual sense with respect to both point and contact transformations, 
i.e., its point and contact equivalence groupoids coincide with the action groupoids of 
its point equivalence group~$G^\sim$ and of the first prolongation of this group, respectively. 
The group~$G^\sim$ is generated by the transformations of the form
\begin{gather}\label{eq:NonlinKGEqsGsim0}
\tilde t=T(t),\quad \tilde x=X(x),\quad \tilde u=Cu+U^0(t,x),\quad \tilde f=\frac{Cf+U^0_{tx}}{T_tX_x}
\end{gather}
and the discrete equivalence transformation $\mathscr I^0$: $\tilde t=x$, $\tilde x=t$, $\tilde u=u$, $\tilde f=f$.
Here $T$, $X$ and $U^0$ are arbitrary smooth functions of their arguments with $T_tX_x\ne0$,
and $C$ is an arbitrary nonzero constant.
\end{lemma}

\begin{corollary}\label{cor:NonlinKGEqsDiscreteEquivTrans}
A complete list of discrete equivalence transformations of the class~\eqref{eq:NonlinKGEqs}
that are independent up to combining with each other and with continuous equivalence transformations of this class
is exhausted by the $(t,x)$-permutation~$\mathscr I^0$ and three transformations alternating signs of variables,
$\mathscr I^t\colon(t,x,u,f)\mapsto (-t,x,u,-f)$,
$\mathscr I^x\colon(t,x,u,f)\mapsto (t,-x,u,-f)$,
$\mathscr I^u\colon(t,x,u,f)\mapsto (t,x,-u,-f)$.
The quotient group of the equivalence group $G^\sim$ of the class~\eqref{eq:NonlinKGEqs}
with respect to its identity component is isomorphic to the group $\mathrm D_4\times\mathbb Z_2$, 
where the dihedral group $\mathrm D_4$ is the symmetry group of a square.
\end{corollary}

\begin{corollary}\label{cor:NonlinKGEqsContactTransToWaveEq}
There is no contact transformation that maps an equation of the form $u_{tx}=f(t,x,u)$ with $f_u\ne0$ 
to equations of the same form with $f_u=0$.
\end{corollary}

\begin{proof}
Suppose that there exists such a contact transformation~$\Phi$. 
Repeating the proof of Lemma~\ref{lem:GenNonlinKGEqsContactAdmTrans} for $f_u\ne0$ and $\tilde f_u=0$, 
we derive that the transformation~$\Phi$ is the first prolongation of a point transformation 
in the space with the coordinates $(t,x,u)$. 
Theorem~4.3c from \cite{king1998a} implies that the (point) equivalence group
of the superclass~$\mathcal K_{\rm gen}$ of all the equation of the form $u_{tx}=f(t,x,u)$
coincides with the group~$G^\sim$, 
and any point admissible transformation within~$\mathcal K_{\rm gen}$ is generated by an element of~$G^\sim$. 
The conditions $f_u\ne0$ and $f_u=0$ are $G^\sim$-invariant, which contradicts the supposition. 
\end{proof}

Analyzing results of~\cite{lie1881b,lie1881a}, one can deduce 
that assertions like Lemmas~\ref{lem:GenNonlinKGEqsContactAdmTrans} and~\ref{lem:NonlinKGEqsGsim} 
and Corollary~\ref{cor:NonlinKGEqsContactTransToWaveEq} may have been known to Sophus Lie. 
In general, similar assertions are typical 
for the theory of contact equivalence of Monge--Amp\`ere equations, 
see, e.g., Lemma~1 in~\cite[p.~205]{kush2010a} and references therein.
In particular, the above Corollary~\ref{cor:NonlinKGEqsContactTransToWaveEq} follows from Corollary~1 in~\cite[p.~238]{kush2009b}. 

\begin{remark}\label{rem:NonlinKGEqsOnContactTrans}
The above assertions imply that the class~$\mathcal K_{\rm gen}$ 
and its subclasses~$\mathcal K_{\rm gen}\setminus\bar{\mathcal K}$, $\bar{\mathcal K}$, $\bar{\mathcal K}\setminus\mathcal K$ and $\mathcal K$, 
which are singled out by the auxiliary constraints $f_u=0$, $f_u\ne0$, $f_u\ne0\wedge f_{uu}=0$ and $f_{uu}\ne0$, respectively,
have the same point equivalence group~$G^\sim$ and are normalized in the point sense. 
The subclasses~$\bar{\mathcal K}$, $\bar{\mathcal K}\setminus\mathcal K$ and $\mathcal K$ 
are also normalized in the contact sense. 
Equations from the class~$\mathcal K$ are not mapped by contact transformations to equations from the class~$\mathcal K_{\rm gen}\setminus\mathcal K$. 
Although the group classification of the class~$\mathcal K_{\rm gen}\setminus\mathcal K$ up to the $G^\sim$-equivalence 
has not be carried out in the literature, Lie's solution of the group classification problem 
for the wider class of all linear hyperbolic equations with two independent variables 
is well known~\cite{lie1881a}. 
Moreover, symmetry properties of the linear and the nonlinear equations from~$\mathcal K_{\rm gen}$, 
which constitute the classes~$\mathcal K_{\rm gen}\setminus\mathcal K$ and~$\mathcal K$, respectively, are quite different. 
The last three facts justify the exclusion of linear equations from the further consideration. 
\end{remark}

Lemma~\ref{lem:NonlinKGEqsGsim} implies that the transformations of the form~\eqref{eq:NonlinKGEqsGsim0} 
constitute a subgroup~$H$ of~$G^\sim$. 
Any such transformation~$\mathscr T$ can be represented as a composition
\[
\mathscr T = \mathscr D^t(T)\circ\mathscr D^x(X)\circ\mathscr Z(U^0)\circ\mathscr D^u(C)
\]
of the elementary equivalence transformations 
\[
\begin{array}{@{}llllll}
\mathscr D^t(T) \colon\ & \tilde t=T(t), \ & \tilde x=x,   \ & \tilde u=u,           \ & \tilde f=f/T_t,\\[1ex]
\mathscr D^x(X) \colon\ & \tilde t=t,    \ & \tilde x=X(x),\ & \tilde u=u,           \ & \tilde f=f/X_x,\\[1ex]
\mathscr D^u(C) \colon\ & \tilde t=t,    \ & \tilde x=x,   \ & \tilde u=Cu,          \ & \tilde f=Cf,\\[1ex]
\mathscr Z(U^0) \colon\ & \tilde t=t,    \ & \tilde x=x,   \ & \tilde u=u+U^0(t, x), \ & \tilde f=f+U^0_{tx},
\end{array}
\]
which are an arbitrary transformation in $t$, an arbitrary transformation in $x$,
a scaling of $u$ and a shift of $u$ with arbitrary functions of $(t,x)$, respectively.
The transformation parameters are described in Lemma~\ref{lem:NonlinKGEqsGsim}, 
and their values are the same as in the form~\eqref{eq:NonlinKGEqsGsim0}. 
The families of elementary transformations 
$\{\mathscr D^t(T)\}$, $\{\mathscr D^x(X)\}$, $\{\mathscr D^u(C)\}$ and $\{\mathscr Z(U^0)\}$, 
where the corresponding constant or functional parameter varies, are subgroups of~$G^\sim$. 
One more elementary equivalence transformation of the class~\eqref{eq:NonlinKGEqs} is~$\mathscr I^0$, 
whereas $\mathscr I^t=\mathscr D^t(-t)$, $\mathscr I^x=\mathscr D^t(-x)$ and $\mathscr I^u=\mathscr D^u(-1)$.
Each transformation~$\mathscr T$ from~$G^\sim\setminus H$ can be decomposed as
\[
\mathscr T = \mathscr I^0\circ\mathscr D^t(T)\circ\mathscr D^x(X)\circ\mathscr Z(U^0)\circ\mathscr D^u(C).
\]

\begin{corollary}
The equivalence algebra of the class~\eqref{eq:NonlinKGEqs} is
$
\mathfrak g^\sim:=\big\langle\hat D^t(\tau),\,\hat D^x(\xi),\,\hat I,\,\hat Z(\eta^0)\big\rangle,
$
where the parameter functions $\tau=\tau(t)$, $\xi=\xi(x)$ and~$\eta^0=\eta^0(t,x)$ run through the sets of smooth functions of their arguments, and
\begin{gather*}
\hat D^t(\tau):=\tau(t)\p_t-\tau_t(t)f\p_f,\quad
\hat D^x(\xi):=\xi(x)\p_x-\xi_x(x)f\p_f,\\
\hat I:=u\p_u+f\p_f,\quad
\hat Z(\eta^0):=\eta^0(t,x)\p_u+\eta^0_{tx}(t,x)\p_f.
\end{gather*}
\end{corollary}

Using the infinitesimal invariance criterion~\cite{blum2010A,blum1989A,olve1993A},
we prove the following assertion.

\begin{proposition}\label{pro:NonlinKGEqsMIA}
The maximal Lie invariance algebra~$\mathfrak g_f$ of an equation~$K_f$ from the class~\eqref{eq:NonlinKGEqs}
consists of the vector fields of the form $\tau(t)\p_t+\xi(x)\p_x+\big(\eta^1u+\eta^0(t,x)\big)\p_u$,
where the parameter functions $\tau=\tau(t)$, $\xi=\xi(x)$ and~$\eta^0=\eta^0(t,x)$
and the constant~$\eta^1$ satisfy the classifying equation
\begin{gather}\label{eq:NonlinKGEqsClassifyingEq}
\tau f_t + \xi f_x + \big(\eta^1u + \eta^0\big)f_u=\big(\eta^1-\tau_t-\xi_x\big)f+\eta^0_{tx}.
\end{gather}
\end{proposition}

Consider the linear span $\mathfrak g_\spanindex$ of all the maximal Lie invariance algebras
of equations from the class~\eqref{eq:NonlinKGEqs},
\begin{gather}\label{eq:NonlinKGEqsSpan}
\begin{split}
\mathfrak g_\spanindex:={}&\textstyle\sum_f\mathfrak g_f=\big\{Q=\tau(t)\p_t+\xi(x)\p_x+(\eta^1u +\eta^0(t,x))\p_u\big\}\\
={}&\textstyle\big\langle D^t(\tau),\,D^x(\xi),\,I,\,Z(\eta^0)\big\rangle
\ne\bigcup_f \mathfrak g_f,
\end{split}
\end{gather}
where the parameters $\tau$, $\xi$ and~$\eta^0$ run through the sets of smooth functions of their arguments,
$\eta^1$ is an arbitrary constant, and
\begin{gather*}
D^t(\tau):=\tau(t)\p_t,\quad
D^x(\xi):=\xi(x)\p_x,\quad
I:=u\p_u,\quad
Z(\eta^0):=\eta^0(t,x)\p_u.
\end{gather*}
It is obvious that any vector field~$Q$ of the above form with $(\tau,\xi)\ne(0,0)$
belongs to~$\mathfrak g_f$ for any~$f$ satisfying the classifying equation with the components of~$Q$,
and such a value of the arbitrary element~$f$ necessarily exists.
The last inequality in~\eqref{eq:NonlinKGEqsSpan} holds since the vector fields from~$\langle I,Z(\eta^0)\rangle$
do not belong to~$\mathfrak g_f$ for any~$f$
in view of the auxiliary inequality $f_{uu}\ne 0$ for the arbitrary element~$f$ within the class~\eqref{eq:NonlinKGEqs}.
At the same time, we can represent such vector fields as linear combinations of vector fields from~$\mathfrak g_\spanindex$
with $(\tau,\xi)\ne(0,0)$.
This is why the second equality in~\eqref{eq:NonlinKGEqsSpan} holds as well,
and thus the algebra $\mathfrak g_\spanindex$ coincides
with the projection $\varpi_*\mathfrak g^\sim$ of~$\mathfrak g^\sim$.
Moreover, the span~$\mathfrak g_\spanindex$ is a Lie algebra, since it is closed with respect to the Lie bracket of vector fields.
\looseness=-1

Here and in what follows $\varpi$ denotes the projection of the space with coordinates $(t,x,u,f)$ onto the space with coordinates $(t,x,u)$. 
We also use the notation $\pi^{t,x}$, $\pi^t$ and~$\pi^x$ 
for the projections of the space with coordinates $(t,x,u)$ onto the spaces with coordinates $(t,x)$, $t$ and~$x$, respectively.

The nonidentity actions of elementary equivalence transformations on the above vector fields spanning $\mathfrak g_\spanindex$
are the following:
\begin{alignat*}{3}
&\big(\varpi_*\mathscr D^t(T)\big)_*D^t(\tau)=D^t\big(\tau(\hat T)/\hat T_t\big), &&
 \big(\varpi_*\mathscr D^t(T)\big)_*Z(\eta^0)=Z\big(\eta^0(\hat T,x)\big),        \\
&\big(\varpi_*\mathscr D^x(X)\big)_*D^x(\xi )=D^x\big(\xi (\hat X)/\hat X_x\big), &&
 \big(\varpi_*\mathscr D^x(X)\big)_*Z(\eta^0)=Z\big(\eta^0(t,\hat X)\big),        \\
&\big(\varpi_*\mathscr Z(U^0)\big)_*D^t(\tau)=D^t(\tau)+Z(\tau U^0_t),            &&
 \big(\varpi_*\mathscr D^u(C)\big)_*Z(\eta^0)=Z(C\eta^0),                         \\
&\big(\varpi_*\mathscr Z(U^0)\big)_*D^x(\xi )=D^x(\xi )+Z(\xi  U^0_x),\qquad      &&
 \big(\varpi_*\mathscr Z(U^0)\big)_*I        =I-Z(U^0),
\end{alignat*}
where
$\hat T=\hat T(t)$ and $\hat X=\hat X(x)$ are the inverses of the functions~$T$ and~$X$, respectively.

\begin{definition}
A subalgebra $\mathfrak s$ of $\mathfrak g_\spanindex$ is said to be appropriate
if there exists a value of the arbitrary element~$f$ such that $\mathfrak s=\mathfrak g_f$.
\end{definition}

In view of Lemma~\ref{lem:NonlinKGEqsGsim}, the group classification of the class~\eqref{eq:NonlinKGEqs}
reduces to the classification of appropriate subalgebras
of $\mathfrak g_\spanindex = \varpi_*\mathfrak g^\sim$ up to the $\varpi_*G^\sim$-equivalence.

Splitting the classifying equation~\eqref{eq:NonlinKGEqsClassifyingEq}
with respect to the arbitrary element~$f$ and its derivations,
we obtain the trivial system $\tau=0$, $\xi=0$, $\eta^1=0$ and $\eta^0=0$.
This means that the following assertion holds.

\begin{lemma}
The kernel Lie invariance algebra of the equations from the class~\eqref{eq:NonlinKGEqs} is $\mathfrak g^\cap=\{0\}$.
\end{lemma}

Since the kernel Lie invariance algebra~$\mathfrak g^\cap$ is zero,
the condition of the necessary inclusion of it into each appropriate algebra
makes no constraint for such algebras.
Analyzing the classifying equation~\eqref{eq:NonlinKGEqsClassifyingEq} deeper,
we derive really essential constraints for such algebras.

\begin{lemma}\label{lem:NonlinKGEqsConditionsForAppropriateSubalgebras}
{\rm (i)} $\mathfrak g_f\cap\big\langle I,\,Z(\eta^0)\big\rangle =\{0\}$ for any $f=f(t,x,u)$ with $f_{uu}\ne0$,
and therefore \mbox{$\dim\mathfrak g_f=\dim\pi^{t,x}_*\mathfrak g_f$}.
Here $\eta^0$ runs through the set of smooth functions depending on~$(t,x)$.

{\rm (ii)} $\dim\mathfrak g_f=\infty$ if and only if $f={\rm e}^u$ $(\!{}\bmod G^\sim)$.

{\rm (iii)} If $f\ne{\rm e}^u$ $(\!{}\bmod G^\sim)$, then $\dim\mathfrak g_f\leqslant4$.
\end{lemma}

\begin{proof}
Suppose that for a value of the arbitrary element~$f$, the algebra~$\mathfrak g_f$ contains a vector field~$\eta^1I+Z(\eta^0)$
with $(\eta^1,\eta^0)\ne(0,0)$.
The classifying equation~\eqref{eq:NonlinKGEqsClassifyingEq} implies that
the function~$f$ satisfies the equation $(\eta^1u+\eta^0)f_u=\eta^1f+\eta^0_{tx}$.
Considering the cases $\eta^1\ne0$ and $\eta^1=0$ separately, we easily show that in both cases
$f$ is affine in~$u$, which contradicts the auxiliary inequality~$f_{uu}\ne0$ for the class~\eqref{eq:NonlinKGEqs}.
This proves item~(i) of the lemma.

We differentiate the classifying equation~\eqref{eq:NonlinKGEqsClassifyingEq} with respect to~$u$
and, in view of the auxiliary inequality~$f_{uu}\ne0$,
divide the result of differentiation by~$f_{uu}$.
Then we differentiate the obtained equation once more with respect to~$u$,
which gives
\begin{gather}\label{eq:NonlinKGEqsClassifyingEqDiffConsequence}
\tau\left(\frac{f_{ut}}{f_{uu}}\right)_u+\xi\left(\frac{f_{ux}}{f_{uu}}\right)_u+\eta^1=-(\tau_t+\xi_x)\left(\frac{f_u}{f_{uu}}\right)_u.
\end{gather}
We need to consider two cases depending on whether or not the expression $(f_u/f_{uu})_u$ vanishes.

Upon the condition $(f_u/f_{uu})_u\ne0$, we can rewrite the equation~\eqref{eq:NonlinKGEqsClassifyingEqDiffConsequence} as
\[
\tau_t+\xi_x=-\tau\frac{(f_{ut}/f_{uu})_u}{({f_u}/{f_{uu}})_u}
-\xi\frac{(f_{ux}/f_{uu})_u}{({f_u}/{f_{uu}})_u}-\eta^1\frac1{({f_u}/{f_{uu}})_u}.
\]
After fixing a value $u=u_0$, the last equation takes the form
$\tau_t+\xi_x=A(t,x)\tau+B(t,x)\xi+C(t,x)$,
where the coefficients~$A$, $B$ and~$C$ are obviously expressed via derivatives of~$f$ at $u=u_0$.
After additionally fixing a value $t=t_0$, we derive the first-order inhomogeneous linear ordinary differential equation
\[\xi_x=B(t_0,x)\xi-\tau_t(t_0)+A(t_0,x)\tau(t_0)+\eta^1C(t_0,x)\]
with respect to~$\xi$,
and its inhomogeneity involves three constant parameters~$\tau(t_0)$, $\tau_t(t_0)$ and~$\eta^1$.
The general solution of this equation can be represented in the form
\[\xi=C_1\xi^1(x)+\tau(t_0)\xi^2(x)+\tau_t(t_0)\xi^3(x)+\eta^1\xi^4(x),\]
where $\xi^k(x)$, $k=1,\dots,4$, are fixed smooth functions of~$x$,
and hence it is linearly parameterized at most four independent arbitrary constants.
In other words, $\dim\pi^x_*\mathfrak g_f\leqslant4$.
Similarly, after taking into account the derived expression for~$\xi$,
we fix a value~$x_0$ for~$x$ instead of~$t$ and obtain the first-order inhomogeneous linear ordinary differential equation
\begin{gather*}
\begin{split}
\tau_t={}&A(t,x_0)\tau+B(t,x_0)\big(C_1\xi^1(x_0)+\tau(t_0)\xi^2(x_0)+\tau_t(t_0)\xi^3(x_0)+\eta^1\xi^4(x_0)\big)\\
&{}-C^1\xi^1_x(x_0)-\tau(t_0)\xi^2_x(x_0)-\tau_t(t_0)\xi^3_x(x_0)-\eta^1\xi^4_x(x_0)+C(t,x_0)
\end{split}
\end{gather*}
with respect to~$\tau$, where the inhomogeneity involves the four constant parameters~$\tau(t_0)$, $\tau_t(t_0)$, $\eta^1$ and~$C_1$.
Since the value of~$\tau$ in the fixed point~$t=t_0$ is among these parameters,
the general solution of this equation merely involves these very parameters,
i.e., $\dim\pi^t_*\mathfrak g_f\leqslant4$.
Therefore, $\dim\pi^{t,x}_*\mathfrak g_f\leqslant4$ and thus, in view of item~(i) of the lemma, $\dim\mathfrak g_f\leqslant4$.

Now we study the second case $(f_u/f_{uu})_u=0$, i.e.,
$f=\gamma(t,x){\rm e}^{\alpha(t,x)u}+\beta(t,x)$, where $\alpha\gamma\ne0$
in view of the auxiliary inequality~$f_{uu}\ne0$ for the class~\eqref{eq:NonlinKGEqs}.
Hence $\gamma=1$ $(\!{}\bmod G^\sim)$.
Substituting the expression for~$f$ into the classifying equation~\eqref{eq:NonlinKGEqsClassifyingEq},
we collect the coefficients of the linearly independent functions $u{\rm e}^{\alpha u}$, ${\rm e}^{\alpha u}$ and~1,
treated as functions of~$u$, which leads to the system
\begin{gather}\label{eq:NonlinKGEqsClassifyingEqDiffConsequence2}
\tau\alpha_t+\xi\alpha_x+\eta^1\alpha=0,\quad
\alpha\eta^0=\eta^1-\tau_t-\xi_x,\quad
\tau\beta_t+\xi\beta_x=\big(\eta^1-\tau_t-\xi_x\big)\beta+\eta^0_{tx}.
\end{gather}
If $\alpha_x\ne0$, then \smash{$\xi=-\big((\tau\alpha_t+\eta^1\alpha)/\alpha_x\big)\big|_{t=t_0}$},
i.e., the component~$\xi$ involves at most two varying constants, $\tau(t_0)$ and~$\eta^1$.
Similarly, if $\alpha_t\ne0$, then \smash{$\tau=\big((\xi\alpha_x+\eta^1\alpha)/\alpha_t\big)\big|_{x=x_0}$},
i.e., the component~$\tau$ also involves at most two varying constants, $\xi(x_0)$ and~$\eta^1$.
Therefore, in the case $\alpha_t\alpha_x\ne0$, the components~$\tau$ and~$\xi$ in total involve
at most three different varying constants, and hence $\dim\mathfrak g_f\leqslant3$.
If exactly one of the derivatives~$\alpha_t$ and~$\alpha_x$ is nonzero,
then up to the equivalence transformation~$\mathscr I^0$, we can assume that $\alpha_x\ne0$ and~$\alpha_t=0$.
Then the substitution of the expressions for $\xi$ and~$\eta^0$ implied by
the first two equations of~\eqref{eq:NonlinKGEqsClassifyingEqDiffConsequence2},
$\xi=-\eta^1\alpha/\alpha_x$ and $\eta^0=(\eta^1-\tau_t-\xi_x)/\alpha$,
into the last equation of~\eqref{eq:NonlinKGEqsClassifyingEqDiffConsequence2}
leads to the second-order inhomogeneous linear ordinary differential equation
$(1/\alpha)_x\tau_{tt}-(\beta\tau)_t=\eta^1(\alpha\beta/\alpha_x)_x-\eta^1\beta$
with respect to~$\tau$,
where the leading coefficient $(1/\alpha)_x$ does not vanish,
and the inhomogeneity involves the single constant parameter~$\eta^1$.
Analogously to the above consideration, we obtain $\dim\mathfrak g_f\leqslant3$.
Otherwise, $\alpha_t=\alpha_x=0$, and hence $\alpha=\const$.
Since $\alpha\ne0$, we can set $\alpha=1$ by scalings of~$u$ and~$\mathscr I^u$.
The above system reduces to $\eta^1=0$, $\eta^0=-\tau_t-\xi_x$, $\tau\beta_t+\xi\beta_x=-(\tau_t+\xi_x)\beta$.
For $\beta\ne0$, we treat the last equation in the same way
as the equation~\eqref{eq:NonlinKGEqsClassifyingEqDiffConsequence}
and conclude that $\dim\mathfrak g_f\leqslant3$ in this case.
If $\beta=0$, then the equation~$K_f$ coincides with the Liouville equation $u_{tx}={\rm e}^u$,
whose maximal Lie invariance algebra is $\mathfrak g_f=\langle\tau(t)\p_t+\xi(x)\p_x-(\tau_t(t)+\xi_x(x))\p_u\rangle$,
where the components~$\tau$ and~$\xi$ run through the sets of smooth functions of~$t$ or~$x$, respectively.
Therefore, $\dim\mathfrak g_f=\infty$,
and up to the $G^\sim$-inequivalence, this is the only case with the infinite dimension of~$\dim\mathfrak g_f$,
which proves item (ii) of the lemma.

For all the other equations from the class~\eqref{eq:NonlinKGEqs}, we derive that
the dimensions of the corresponding maximal Lie invariance algebras do not exceed four,
which implies item (iii) of the lemma.
\end{proof}

\begin{corollary}\label{cor:NonlinKGEqsProjectionDim}
If $f\ne{\rm e}^u$ $(\!{}\bmod G^\sim)$,
then $\dim\pi^t_*\mathfrak g_f\leqslant3$ and $\dim\pi^x_*\mathfrak g_f\leqslant3$.
\end{corollary}

\begin{proof}
In view of item (iii) of Lemma~\ref{lem:NonlinKGEqsConditionsForAppropriateSubalgebras},
we have $\dim\mathfrak g_f\leqslant4$ if $f\ne{\rm e}^u$ $(\!{}\bmod G^\sim)$.
Then $\dim\pi^t_*\mathfrak g_f\leqslant\dim\mathfrak g_f\leqslant4$ and $\dim\pi^x_*\mathfrak g_f\leqslant\dim\mathfrak g_f\leqslant4$,
i.e., $\pi^t_*\mathfrak g_f$ and~$\pi^x_*\mathfrak g_f$
are finite-dimensional Lie algebras of vector fields on the $t$- and the $x$-lines, respectively.
Then the required inequalities directly follow from the Lie theorem on such algebras.
\end{proof}

\section{Result of group classification}
\label{sec:NonlinKGEqsResults}

The main result of the paper is the following theorem.

\begin{theorem}\label{thm:NonlinKGEqsGroupClassification}
A complete list of $G^\sim$-inequivalent (maximal) Lie-symmetry extensions
in the class~\eqref{eq:NonlinKGEqs} is exhausted by the following cases:
{\rm
\begin{enumerate}\itemsep=0ex\setcounter{enumi}{-1}
\item\label{EKGcase0} 
General case $f=f(t,x,u)$: \ $\{0\}$;
\item\label{EKGcase1} 
$f=\hat f(x,u)$: \ $\langle\p_t\rangle$;
\item\label{EKGcase2} 
$f=\hat f(x-t,u)$: \ $\langle\p_t+\p_x\rangle$;
\item\label{EKGcase3} 
$f={\rm e}^t\hat f(x,{\rm e}^{-t}u)$: \ $\langle\p_t+u\p_u\rangle$;
\item\label{EKGcase4} 
$f={\rm e}^{x+t}\hat f(x-t,{\rm e}^{-x-t}u)$: \ $\langle\p_t+\p_x+2u\p_u\rangle$;
\item\label{EKGcase5} 
$f={\rm e}^t\hat f({\rm e}^{-t}u)$: \ $\langle\p_t+u\p_u,\,\p_x\rangle$;
\item\label{EKGcase6} 
$f={\rm e}^{x+t}\hat f({\rm e}^{-x-t}u)$: \ $\langle\p_t+u\p_u,\,\p_x+u\p_u\rangle$;
\item\label{EKGcase7} 
$f=|x-t|^{-q-2}\hat f(|x-t|^qu)$, $q\ne0$: \ $\langle\p_t+\p_x,\,t\p_t+x\p_x-qu\p_u\rangle$;
\item\label{EKGcase8} 
$f=|x|^{-q-2}\hat f(|x|^qu)$, $q\ne0$: \ $\langle\p_t,\,t\p_t+x\p_x-qu\p_u\rangle$;
\item\label{EKGcase9} 
$f=\hat f(u)$: \ $\langle\p_t,\,\p_x,\,t\p_t-x\p_x\rangle$;
\item\label{EKGcase10} 
$f=(x-t)^{-2}\hat f(u)$: \ $\langle\p_t+\p_x,\,t\p_t+x\p_x,\,t^2\p_t+x^2\p_x\rangle$;
\item\label{EKGcase11} 
$f={\rm e}^{u/x}$: \ $\langle\p_t,\,t\p_t-x\p_u,\,x\p_x+u\p_u\rangle$;
\item\label{EKGcase12} 
$f=|u|^p u$, \ $p\ne-1,0$: \ $\langle\p_t,\,\p_x,\,t\p_t-x\p_x,\,-pt\p_t+u\p_u\rangle$;
\item\label{EKGcase13} 
$f={\rm e}^u$: \ $\langle\tau(t)\p_t+\xi(x)\p_x-(\tau_t(t)+\xi_x(x))\p_u\rangle$.
\end{enumerate}
}
\noindent
Here $\hat f$ is an arbitrary smooth function of its arguments
whose second derivative with respect to the argument involving~$u$ is nonzero,
and $q$ and~$p$ are arbitrary constants that satisfy the conditions indicated in the corresponding cases.
In Case~\ref{EKGcase13}, the components~$\tau$ and~$\xi$
run through the sets of smooth functions of~$t$ or~$x$, respectively.
\end{theorem}

The proof of Theorem~\ref{thm:NonlinKGEqsGroupClassification} is presented in the next Section~\ref{sec:NonlinKGEqsProof}.

\begin{remark}\label{rem:NonlinKGEqsOnMaxLieSymExtensions}
There are two ways of interpreting the classification cases listed in Theorem~\ref{thm:NonlinKGEqsGroupClassification},
in a \emph{weak} sense and in a \emph{strong} sense.
Within the framework of the weak group classification,
we consider the entire subclass~$\mathcal K_N$
of the equations from the class~$\mathcal K$ with the form of~$f$ presented in Case~$N$,
\[
N\in\Gamma:=\{\ref{EKGcase0},\dots,\ref{EKGcase6},\ref{EKGcase7}_q,\ref{EKGcase8}_q,\ref{EKGcase9},
\ref{EKGcase10},\ref{EKGcase11},\ref{EKGcase12}_p,\ref{EKGcase13}\mid q\ne0,\, p\ne-1,0\},
\]
and then the corresponding algebra is the kernel Lie invariance algebra~$\mathfrak g^\cap_N$
of the equations from the subclass~$\mathcal K_N$.
Here we use the notation $\ref{EKGcase7}_q:=(\ref{EKGcase7},q)$, $\ref{EKGcase8}_q:=(\ref{EKGcase8},q)$ and $\ref{EKGcase12}_p:=(\ref{EKGcase12},p)$.
We refer to Cases~\ref{EKGcase7}, \ref{EKGcase8} and~\ref{EKGcase12} as to collections of
Cases~$\ref{EKGcase7}_q$, $\ref{EKGcase8}_q$ and $\ref{EKGcase12}_p$ with fixed values of~$q$ or~$p$, respectively.
It is obvious that $\mathcal K_0=\mathcal K$.
Under the strong group classification,
Case~$N$ includes only the equations from the subclass~$\mathcal K_N$ for which $\mathfrak g_f=\mathfrak g^\cap_N$.
Thus, the discussion after the equation~\eqref{eq:NonlinKGEqsSpan} implies that $\mathfrak g_f=\{0\}$
and hence $K_f$ belongs to strong Case~\ref{EKGcase0}
if and only if $f$ does not satisfy the classifying equation~\eqref{eq:NonlinKGEqsClassifyingEq}
for any constant~$\eta^1$ and
any smooth functions $\tau=\tau(t)$, $\xi=\xi(x)$ and~$\eta^0=\eta^0(t,x)$ with $(\tau,\xi)\ne(0,0)$.
For \mbox{Cases~\ref{EKGcase1}--\ref{EKGcase9}} to merely collect maximal Lie-symmetry extensions,
the parameter function~$\hat f$ should take only values
for which the associated values of the arbitrary element~$f$
are not $G^\sim$-equivalent to ones from the other listed cases
with maximal Lie invariance algebras of greater dimensions.
In other words, a value of the parameter function~$\hat f$ leads to a maximal Lie-symmetry extension
if and only if it satisfies no equation among those associated with the corresponding case in
Proposition~\ref{pro:FurtherLieSymExtestionsForCases1-4} or~\ref{pro:FurtherLieSymExtestionsForCases5-9} below.
Case~\ref{EKGcase10} is special since $\mathfrak g_f=\mathfrak g^\cap_{10}$ for any $K_f\in\mathcal K_{10}$,
and thus it does not depend on interpreting the group classification.
We will mostly omit the attributes ``weak'' and ``strong'', explicitly indicating all places where the weak interpretation is used.
\end{remark}

\begin{remark}\label{rem:NonlinKGEqsOnEquivFormOfClassificationCases}
Cases~\ref{EKGcase3}--\ref{EKGcase6} and~\ref{EKGcase8} can be replaced by $G^\sim$-equivalent cases,
for each of which the arbitrary element~$f$ just runs through the set of arbitrary smooth functions of either one or two arguments
without an additional multiplier:
\begin{gather*}
3'.\ f=\hat f\big(x,t^{-1}u\big)\colon \ \ \mathfrak g_{f}=\langle t\p_t+u\p_u\rangle;\\[.5ex]
4'.\ f= \hat f\big(t^{-1}x,(tx)^{-1}u\big)\colon \ \ \mathfrak g_f=\langle t\p_t+x\p_x+2u\p_u\rangle;\\[.5ex]
5'.\ f= \hat f\big(t^{-1}u\big)\colon \ \ \mathfrak g_f=\langle t\p_t+u\p_u,\,\p_x\rangle;\\[.5ex]
6'.\ f= \hat f\big((tx)^{-1}u\big)\colon \ \ \mathfrak g_f=\langle t\p_t+u\p_u,\,x\p_x+u\p_u\rangle;\\[.5ex]
8'a.\ f= \hat f\big(|x|^{q'}u\big), \ q'\ne0,-1\colon \ \ \mathfrak g_f=\langle\p_t,\,(q'+1)t\p_t-x\p_x+q'u\p_u\rangle;\\[.5ex]
8'b.\ f= \hat f\big({\rm e}^{-x}u\big)\colon \ \ \mathfrak g_f=\langle\p_t,\,t\p_t+\p_x+u\p_u\rangle.
\end{gather*}
Here Case~\ref{EKGcase8} splits into two subcases, 8$'$a and~8$'$b,
respectively associated with $q\ne0,-1$ and $q=-1$.
\end{remark}

\begin{remark}\label{rem:NonlinKGEqsOnTransitionToStandardNonlinKGEqs}
The group classification of the class~\eqref{eq:NonlinKGEqs} can be easily mapped
via the point transformation $\check t=x+t$, $\check x=x-t$, $\check u=u$, $\check f=f$
to the group classification of nonlinear Klein--Gordon equations in the standard spacetime variables,
$\check u_{\check t\check t}-\check u_{\check x\check x}=\check f(\check t,\check x,\check u)$.
\end{remark}

\begin{remark}\label{rem:NonlinKGEqsOnMinTupleOfDistingueshingInvIntegers}
We found eight triples of $G^\sim$-invariant integer characteristics of subalgebras~$\mathfrak s$ of~$\mathfrak g_\spanindex$,
which suffice for distinguishing the $G^\sim$-inequivalent cases of Lie-symmetry extensions from each other,
see Remark~\ref{rem:IdentifyingInvValues} below.
The most significant among these eight triples is $(r_3,j_1,r_2)$,
where
\begin{gather*}
r_3=r_3(\mathfrak s):=3-\min\big\{\dim\big\langle D^t(\tau),D^x(\xi )\big\rangle\mid \exists\,\eta^0\colon D^t(\tau)+D^x(\xi )+I+Z(\eta^0)\in\mathfrak s\big\},\\[.5ex]
j_1=j_1(\mathfrak s):=\max(\dim\mathfrak s^1,\dim\mathfrak s^2),\\[.5ex]
r_2=r_2(\mathfrak s):=\min(\dim\pi^t_*\mathfrak s^{12},\dim\pi^x_*\mathfrak s^{12})
\end{gather*}
with 
$\mathfrak s^1   :=\mathfrak s\cap\big\langle D^t(\tau),Z(\eta^0)\big\rangle$,
$\mathfrak s^2   :=\mathfrak s\cap\big\langle D^x(\xi ),Z(\eta^0)\big\rangle$ and 
$\mathfrak s^{12}:=\mathfrak s\cap\big\langle D^t(\tau),D^x(\xi),Z(\eta^0)\big\rangle$.
\end{remark}

In the next four remarks, we discuss weak Lie-symmetry extensions, omitting the attribute ``weak''. 

\begin{remark}\label{rem:NonlinKGEqsOnOrderingSuccessiveLieSymExtensions}
We can partially order the collection of Lie-symmetry extensions within the class~\eqref{eq:NonlinKGEqs}.
Here $\mbox{Case\ }N\prec\mbox{Case\ }\bar N$ with $N,\bar N\in\Gamma$ means
that Case~$\bar N$ is a further Lie-symmetry extension of Case~$N$ modulo the $G^\sim$-equivalence,
i.e., there exists $\mathscr T\in G^\sim$ such that $\mathfrak s\varsubsetneq(\varpi_*\mathscr T)\bar{\mathfrak s}$,
where $\mathfrak s$ and~$\bar{\mathfrak s}$ are subalgebras of~$\mathfrak g_\spanindex$
associated with Cases~$N$ and~$\bar N$, respectively.
For the corresponding subclasses~$\mathcal K_N$ and~$\mathcal K_{\bar N}$,
we have the inverse inclusion, $\mathcal K_N\varsupsetneq(\varpi_*\mathscr T)\mathcal K_{\bar N}$
with the same $\mathscr T\in G^\sim$.
\end{remark}

\begin{remark}\label{rem:NonlinKGEqsOnMinTupleOfDistinguishingSuccessiveLieSymExtensions}
We assign two more $G^\sim$-invariant integer characteristics of subalgebras~$\mathfrak s$ of~$\mathfrak g_\spanindex$
to those from Remark~\ref{rem:NonlinKGEqsOnMinTupleOfDistingueshingInvIntegers},
\[
n=n(\mathfrak s):=\dim\mathfrak s,\quad
k=k(\mathfrak s):=\min(\dim\pi^t_*\mathfrak s,\dim\pi^x_*\mathfrak s),
\]
see Remark~\ref{rem:IdentifyingInvValues} below.
Thus, we consider tuples of the form $(n,r_3,r_2,j_1,k)$,
where the entries are ordered according to their importance,
see the beginning of Section~\ref{sec:SuccessiveLie-SymExtensions}.
We partially order the set of these tuples,
assuming $(n,r_3,r_2,j_1,k)<(\bar n,\bar r_3,\bar r_2,\bar j_1,\bar k)$
if $n<\bar n$, $r_3\leqslant\bar r_3$, $r_2\leqslant\bar r_2$, $j_1\leqslant\bar j_1$ and $k\leqslant\bar k$.
Let $(n,r_3,r_2,j_1,k)$ and $(\bar n,\bar r_3,\bar r_2,\bar j_1,\bar k)$
correspond to the algebras~$\mathfrak s$ and~$\bar{\mathfrak s}$ of Cases~$N$ and~$\bar N$, respectively.
Then the relation $\mbox{Case\ }N\prec\mbox{Case\ }\bar N$ is equivalent
to the relation $(n,r_3,r_2,j_1,k)<(\bar n,\bar r_3,\bar r_2,\bar j_1,\bar k)$.
\end{remark}

\begin{remark}\label{rem:NonlinKGEqsOnHasseDiagram}   
The Hasse diagram for the partially ordered set of $G^\sim$-inequivalent Lie-symmetry extensions
within the class~\eqref{eq:NonlinKGEqs} 
is presented in Figure~\ref{fig:HasseDiagram}.
According to the rule of constructing Hasse diagrams,
the arrows in Figure~\ref{fig:HasseDiagram} depict only the direct Lie-symmetry extensions.
We say that the pair $(\mbox{Case\ }N,\mbox{Case\ }\bar N)$ of classification cases
with $\mbox{Case\ }N\prec\mbox{Case\ }\bar N$
presents a \emph{direct Lie-symmetry extension} if there does not exist $\check N\in\Gamma$ such that
$
\mbox{Case\ }N\prec\mbox{Case\ }\check N\prec\mbox{Case\ }\bar N.
$
\begin{figure}[t!]
\centering
\includegraphics[width=.8\linewidth]{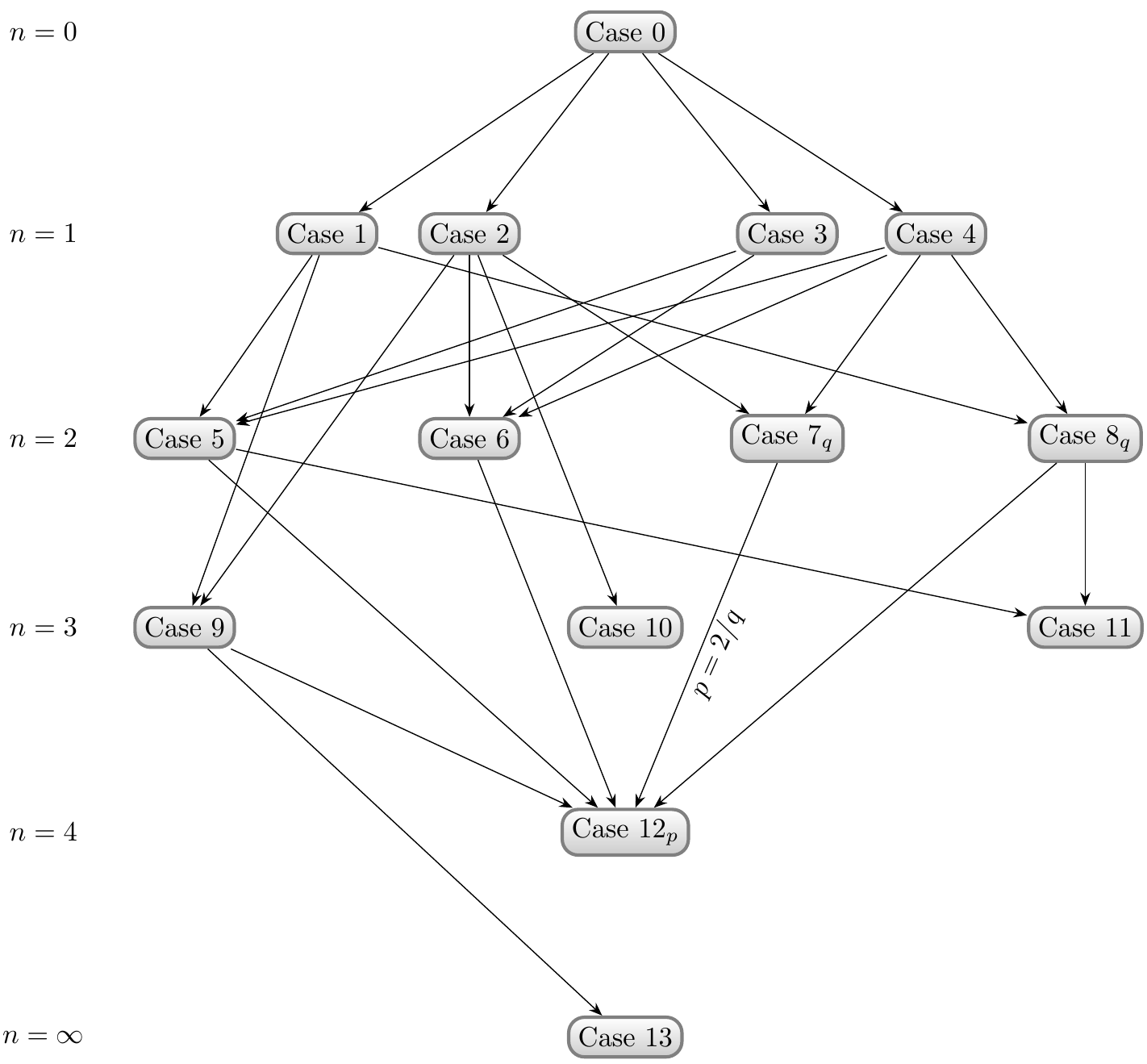}
\caption{The Hasse diagram of Lie-symmetry extensions within the class~\eqref{eq:NonlinKGEqs}.}\label{fig:HasseDiagram}
\end{figure}
\end{remark}

\begin{remark}\label{rem:NonlinKGEqsOnLimitProcessesBetweenClassificationCases}
There are many pairs among $\{(\mbox{Case\ }N,\mbox{Case\ }\bar N),\,N,\bar N\in\Gamma\}$, 
where the pair components are related to each other via limit processes 
supplemented, if necessary, with preliminary or subsequent equivalence transformations.%
\footnote{%
Such a limit process was considered for the class of nonlinear diffusion equations 
in~\cite{blum1988a} and~\cite[p.~181]{blum1989A}, where the exponential nonlinearity was excluded 
from the classification list as a limit of power nonlinearities. 
A theory of such limit processes and a number of their examples were presented in \cite{ivan2007c,popo2005a,vane2012b}.
}
All these limits lead to contractions $\mathfrak g^\cap_N\to\mathfrak h\subseteq\mathfrak g^\cap_{\bar N}$ 
as contractions of realizations of abstract Lie algebras by vector fields, 
where $\mathfrak h$ is a subalgebra of~$\mathfrak g^\cap_{\bar N}$, 
which often coincides with the entire~$\mathfrak g^\cap_{\bar N}$. 
The most obvious limit process is 
\[
\mbox{Case~\ref{EKGcase7}}_q\to\mbox{Case~\ref{EKGcase10}}\quad\mbox{as}\quad q\to0 
\quad\mbox{with}\quad \mathfrak g^\cap_{7,1}\simeq\mathfrak g^\cap_{7,q}\to\langle\p_t+\p_x,\,t\p_t+x\p_x\rangle\subsetneq\mathfrak g^\cap_{10}.
\]
An example with a necessary subsequent equivalence transformation is 
\begin{gather*}
\mbox{Case~\ref{EKGcase8}}_q\to\mathscr D^x(-x^{-1})_*(\mbox{Case~\ref{EKGcase9}})\quad\mbox{as}\quad q\to0 \\
\qquad\mbox{with}\quad \mathfrak g^\cap_{8,1}\simeq\mathfrak g^\cap_{8,q}\to\langle\p_t,\,t\p_t+x\p_x\rangle\subsetneq\big(\varpi_*D^x(-x^{-1})\big)_*\mathfrak g^\cap_9.
\end{gather*}
The limit processes 
\begin{gather*}
\mbox{Case~\ref{EKGcase2}}\to\mbox{Case~\ref{EKGcase1}},\quad
\mbox{Case~\ref{EKGcase4}}\to\mbox{Case~\ref{EKGcase2}},\quad
\mbox{Case~\ref{EKGcase6}}\to\mbox{Case~\ref{EKGcase5}},\quad
\mbox{Case~\ref{EKGcase7}}\to\mbox{Case~\ref{EKGcase8}},\\
\mbox{Case~\ref{EKGcase10}}\to\mathscr D^x(-x^{-1})_*(\mbox{Case~\ref{EKGcase9}})
\end{gather*}
as $q\to 0$
with contractions between the entire corresponding Lie algebras
are realized via preliminarily introducing the parameter~$q$ using a scaling equivalence transformation. 
For instance, for each value of the parameter function~$\hat f$ in Case~\ref{EKGcase2}, 
we take the family $f^q=q^{-1}\hat f(x-t,u)$ with $q\ne0$ and act on each equation~$K_{f^q}$ by 
the equivalence transformation $\mathscr D^t(q^{-1}t)$. 
For each $q\ne0$, we obtain the equation~$K_{\tilde f^q}$ with $\tilde f^q=\hat f(\tilde x-q\tilde t,\tilde u)$, 
which is invariant with respect to the algebra $\langle\p_{\tilde t}+q\p_{\tilde x}\rangle$.
The limit as $q\to 0$ then gives Case~\ref{EKGcase1}.
One more set of limit processes 
\[
\mbox{Case~\ref{EKGcase7}}\to\mbox{Case~\ref{EKGcase6}},\quad
\mbox{Case~\ref{EKGcase8}}\to(\mathscr I^t\circ\mathscr I^0)_*(\mbox{Case~\ref{EKGcase5}}),\quad
\mbox{Case~\ref{EKGcase12}}\to\mbox{Case~\ref{EKGcase13}}
\]
is based on the remarkable limit $(1+q^{-1})^q\to{\rm e}$ as $q\to+\infty$.
The first two limit processes again give the contractions between the entire corresponding Lie algebras, 
whereas the last limit process results in a contraction to a proper subalgebra.
We describe in detail the second limit process.
For each value of the parameter function~$\hat f$ in Case~\ref{EKGcase8}, 
we act on the equation~$K_{f^q}$ with $f^q=q|x|^{-q-2}\hat f(|x|^qu)$, $q\in\mathbb R_{\ne0}$, 
by the  equivalence transformation $\mathscr D^x\big(q(x-1)\big)$, 
which leads to the equation~\smash{$K_{\tilde f^q}$} with $\tilde f^q=|1+\tilde x/q|^{-q-2}\hat f(|1+\tilde x/q|^q\tilde u)$.
The Lie invariance algebra of~\smash{$K_{\tilde f^q}$} is $\langle q^{-1}\tilde t\p_{\tilde t}+(1+q^{-1}\tilde x)\p_{\tilde x}-\tilde u\p_{\tilde u}\rangle$.
Proceeding to the limit as $q\to+\infty$, we get $(\mathscr I^t\circ\mathscr I^0)_*(\mbox{Case~\ref{EKGcase5}})$.
A~less standard limit process is $\mbox{Case~\ref{EKGcase11}}\to\mbox{Case~\ref{EKGcase13}}$, 
where we introduce the parameter~$q$ via the transformation $\mathscr D^t(qt)\circ\mathscr D^x\big(q^{-1}(x-1)\big)$
and take the limit as $q\to 0$. 
The contracted algebra $\langle\p_t,\,t\p_t-\p_u,\,\p_x\rangle$ is a (proper) subalgebra 
of the infinite-dimensional algebra~$\mathfrak g^\cap_{13}$.
\end{remark}

\section{Proof of group classification}\label{sec:NonlinKGEqsProof}

In view of Lemma~\ref{lem:NonlinKGEqsConditionsForAppropriateSubalgebras}(ii),
we can exclude Case~\ref{EKGcase13}, which corresponds to the Liouville equation,
from the further consideration.
For convenience, denote by~$\mathcal L$ the subclass of~$\mathcal K$ consisting of the equations
that are $G^\sim$-equivalent to the Liouville equation.
The equations from the class~$\mathcal K$
with finite-dimensional Lie invariance algebras, for which $f\ne{\rm e}^u$ $(\!{}\bmod G^\sim)$,
constitute the complement $\mathcal K\setminus\mathcal L$.
In this notation, we need to classify only Lie symmetries of equations from the class $\mathcal K\setminus\mathcal L$.
The classification of such equations splits into different cases depending on the following
three $G^\sim$-invariant integer values for subalgebras~$\mathfrak s$ of~$\mathfrak g_\spanindex$:
\begin{gather}\label{eq:NonlinKGEqsMNK}
m:=\max(\dim\pi^t_*\mathfrak s,\dim\pi^x_*\mathfrak s),\quad
n:=\dim\mathfrak s,\quad
k:=\min(\dim\pi^t_*\mathfrak s,\dim\pi^x_*\mathfrak s),
\end{gather}
which are defined by~$f$ for $\mathfrak s=\mathfrak g_f$ 
and are listed in accordance with their influence on the classification splitting.
In view of their definition, they satisfy the inequality $0\leqslant n-m\leqslant k\leqslant m\leqslant n$.
According to Lemma~\ref{lem:NonlinKGEqsConditionsForAppropriateSubalgebras}(iii),
we can suppose from the very beginning that $n\leqslant4$,
whereas Corollary~\ref{cor:NonlinKGEqsProjectionDim} implies that $m\leqslant3$.
Moreover, it follows from Lemma~\ref{lem:NonlinKGEqsConditionsForAppropriateSubalgebras}(i) that $m>0$
for any equation in the class~\eqref{eq:NonlinKGEqs} with Lie-symmetry extension.
In other words, looking for \mbox{$G^\sim$-inequivalent} equations~$K_f$ from $\mathcal K\setminus\mathcal L$
with nonzero maximal Lie invariance algebra~$\mathfrak g_f$,
we select the appropriate triples $(m,n,k)$ among those that satisfy the condition
\[
m\in\{1,2,3\},\quad
n\in\{1,2,3,4\},\quad
0\leqslant n-m\leqslant k\leqslant m\leqslant n.
\]

Since the class~\eqref{eq:NonlinKGEqs} possesses the discrete equivalence transformation $\mathscr I^0$,
which permutes the variables~$t$ and~$x$,
without loss of generality we can assume that $m=\dim\pi^t_*\mathfrak g_f$.
We choose an initial basis of~$\mathfrak g_f$ consisting of vector~fields
\[
Q^i=\tau^i(t)\p_t+\xi^i(x)\p_x+\big(\eta^{1i}u +\eta^{0i}(t,x)\big)\p_u, \quad i=1,\dots,n \quad\mbox{with}\quad n\leqslant4,
\]
see Proposition~\ref{pro:NonlinKGEqsMIA}.
In addition to Corollary~\ref{cor:NonlinKGEqsProjectionDim},
the application of the Lie theorem to the group classification of the class~\eqref{eq:NonlinKGEqs}
is based on the fact that $\pi^t_*G^\sim$ and $\pi^x_*G^\sim$ coincide with the diffeomorphism groups
on~the $t$- and the $x$-lines, respectively.

We separately consider the cases with different values of $m\in\{1,2,3\}$ in descending order.

\medskip\par\noindent{$\boldsymbol{m=3.}$}
Therefore, $n\in\{3,4\}$.
We change the basis $(Q^1,\dots,Q^n)$ in such a way that
$\tau^1$, $\tau^2$ and~$\tau^3$ are linearly independent and, if $n=4$, $\tau^4=0$.
Then $\pi^t_*\mathfrak g_f=\langle\tau^1\p_t,\tau^2\p_t,\tau^3\p_t\rangle$
is a~faithful realization of the algebra $\mathfrak{sl}(2,\mathbb R)$ on the line.
This fact has several implications.
In particular,  
the algebra $\mathfrak f:=\langle Q^1,Q^2,Q^3\rangle$ is, 
up to combining $Q^1$, $Q^2$ and~$Q^3$ with~$Q^4$ if $n=4$,
a~faithful realization of $\mathfrak{sl}(2,\mathbb R)$ as well,
and $\pi^x_*\mathfrak f=\langle\xi^1\p_x,\xi^2\p_x,\xi^3\p_x\rangle$ is also a (not necessarily faithful)
realization of $\mathfrak{sl}(2,\mathbb R)$.%
\footnote{%
Indeed, this claim is obvious for $n=3$. 
Suppose that $n=4$.
The algebra~$\mathfrak g_f$ is not solvable 
since otherwise both the algebra~$\mathfrak f$ and  
the projection $\pi^t_*\mathfrak g_f=\pi^t_*\mathfrak f$ are solvable, 
which is not the case. 
There are only two four-dimensional unsolvable real Lie algebras, 
$\mathfrak{sl}(2,\mathbb R)\oplus\mathfrak a$ and $\mathfrak{so}(3)\oplus\mathfrak a$, 
where $\mathfrak a$ is the one-dimensional abelian Lie algebra. 
The algebra~$\mathfrak g_f$ cannot be isomorphic to $\mathfrak{so}(3)\oplus\mathfrak a$ 
in view of $\pi^t_*\mathfrak g_f\simeq\mathfrak{sl}(2,\mathbb R)$. 
Thus, $\mathfrak g_f\simeq\mathfrak{sl}(2,\mathbb R)\oplus\mathfrak a$, 
and the existence of required basis is obvious.  
}
The last realization should be either faithful or zero
since the kernel of any homomorphism of a Lie algebra $\mathfrak g$ to a Lie algebra is an ideal of $\mathfrak g$,
and the algebra $\mathfrak{sl}(2,\mathbb R)$ has no proper ideals.
In other words, either $\xi^1=\xi^2=\xi^3=0$ or $\xi^1\xi^2\xi^3\ne 0$.
In view of the Lie theorem, up to linearly combining $Q^1$, $Q^2$ and~$Q^3$ with each other,
we can set $(\tau^1,\tau^2,\tau^3)=(1,t,t^2)$ $(\!{}\bmod \pi^t_*G^\sim)$.
Moreover, the derived algebra $[\mathfrak f,\mathfrak f]$ of the algebra~$\mathfrak f$ coincides with the algebra~$\mathfrak f$ itself,
and thus $\eta^{1i}=0$, $i=1,2,3$.

In the case $\xi^1=\xi^2=\xi^3=0$, we have $Q^i=t^{i-1}\p_t+\eta^{0i}\p_u$, $i=1,2,3$,
with $\eta^{01}=0$ $(\!{}\bmod G^\sim)$ and, if $n=4$, $\tau^4=0$.
The commutator relations of~$\mathfrak f$,
\begin{gather*}
[Q^1,Q^2]=\p_t+\eta^{02}_t\p_u=Q^1, \\
[Q^1,Q^3]=2t\p_t+\eta^{03}_t\p_u=2Q^2,\\
[Q^2,Q^3]=t^2\p_t+(t\eta^{03}_t-t^2\eta^{02}_t)\p_u=Q^3,
\end{gather*}
imply the system $\eta_t^{02}=0$, $\eta^{03}_t=2\eta^{02}$, $t\eta^{03}_t=\eta^{03}$,
and thus $Q^1=\p_t$, $Q^2=t\p_t-\mu(x)\p_u$, $Q^3=t^2\p_t-2t\mu (x)\p_u$,
where $\mu$ is a smooth function of~$x$.
Successively substituting the obtained vector fields $Q^1$, $Q^2$ and $Q^3$
into the classifying equation~\eqref{eq:NonlinKGEqsClassifyingEq},
we derive the following system with respect to the arbitrary element~$f$:
$f_t=0$, $\mu f_u=f$, $2\mu tf_u=2tf+\mu_x$.
An obvious consequence of this system is $\mu_x=0$ and,
since~$f\ne0$, we also get $\mu\ne0$ and $f=\nu(x){\rm e}^{u/\mu}$.
Therefore, the equation~$K_f$ is $G^\sim$-equivalent to the Liouville equation,
cf.\ the proof of Lemma~\ref{lem:NonlinKGEqsConditionsForAppropriateSubalgebras},
i.e., $\dim\mathfrak g_f=\infty$, which contradicts the supposition $\dim\mathfrak g_f\leqslant4$.

Now we consider the case $\xi^1\xi^2\xi^3\ne 0$ and, up to the $G^\sim$-equivalence, set $\xi^1=1$ and $\eta^{01}=0$.
After expanding the commutator relations of~$\mathfrak f$ and collecting the components of vector fields,
\begin{gather*}
 [Q^1, Q^2]=\p_t+\xi^2_x\p_x+(\eta^{02}_t+\eta_x^{02})\p_u=Q^1, \\
 [Q^1, Q^3]=2t\p_t+\xi^3_x\p_x+(\eta^{03}_t+\eta^{03}_x)\p_u=2Q^2, \\
 [Q^2, Q^3]=t^2\p_t+(\xi^2\xi^3_x-\xi^3\xi^2_x)\p_x +(t\eta^{03}_t+\xi^2\eta^{03}_x-t^2\eta^{02}_t-\xi^3\eta^{02}_x)\p_u=Q^3,
\end{gather*}
we obtain the system
\begin{gather*}
\xi_x^2=1,\quad \eta_t^{02}+\eta_x^{02}=0,\quad
\xi^3_x=2\xi^2,\quad \eta_t^{03}+\eta_x^{03}=2\eta^{02},\\
\xi^2\xi^3_x-\xi^3\xi^2_x=\xi^3,\quad 
t\eta^{03}_t+\xi^2\eta^{03}_x-t^2\eta^{02}_t-\xi^3\eta^{02}_x=\eta^{03}.
\end{gather*}
From the first, the third and the fifth equations of this system,
up to the equivalence transformations of shifts with respect to~$x$,
we derive $\xi^2=x$ and $\xi^3=x^2$.
From the second and the fourth equations, we find
$\eta^{02}=\rho (\omega^-)$ and $\eta^{03}=\rho(\omega^-)\omega^+ +\theta(\omega^-)$
with $\omega^-:=x-t$ and $\omega^+:=x+t$.
Then the sixth equation of the system reduces to $\omega^-\theta_{\omega^-}=\theta$,
which integrates to $\theta=\lambda\omega^-$, where $\lambda$ is an arbitrary constant.
Thus, 
$Q^1=\p_t+\p_x$, $Q^2=t\p_t+x\p_x+\rho(\omega^-)\p_u$, $Q^3=t^2\p_t +x^2\p_x+\big((x+t)\rho(\omega^-)+\lambda(x-t)\big)\p_u$,
where we can set the parameter function $\rho=\rho(\omega^-)$ to zero by the transformation
$\tilde t=t$, $\tilde x=x$, $\tilde u=u-\int(\omega^-)^{-1}\rho(\omega^-)\,{\rm d}\omega^-$,
which is the projection of an equivalence transformation.
We successively substitute the components of the vector fields $Q^1$, $Q^2$ and $Q^3$
into the classifying equation~\eqref{eq:NonlinKGEqsClassifyingEq}.
The derived system with respect to the arbitrary element~$f$,
\[
f_t+f_x=0,\quad tf_t+xf_x=-2f,\quad t^2f_t+x^2f_x+\lambda(x-t)f_u=-2(x+t)f,
\]
is consistent under the inequality $f_u\ne0$ if and only if $\lambda=0$.

Suppose that $n=4$, $\tau^4=0$ and $\xi^4\ne0$.
Since $\dim\pi^x_*\mathfrak g_f\leqslant3$, the component~$\xi^4$
should belong to $\langle\xi^1,\xi^2,\xi^3\rangle=\langle1,x,x^2\rangle$.
Moreover, $[\mathfrak f,\langle Q^4\rangle]\subseteq\langle Q^4\rangle$,
which implies $[\pi^x_*\mathfrak f,\langle\xi^4\p_x\rangle]\subseteq\langle\xi^4\p_x\rangle$,
i.e., $\langle \xi^4\p_x \rangle$ is an ideal of $\pi^x_*\mathfrak f\simeq\mathfrak{sl}(2,\mathbb R)$.
The algebra~$\mathfrak{sl}(2,\mathbb R)$ has no proper ideals.
Hence $\xi^4=0$, which is a contradiction. 
This gives Case~\ref{EKGcase10} of Theorem~\ref{thm:NonlinKGEqsGroupClassification}.
Moreover, Case~\ref{EKGcase10} admits no further Lie-symmetry extensions
via specifying~$\hat f$;
cf.\ the beginning of Section~\ref{sec:SuccessiveLie-SymExtensions}.

\medskip\par\noindent{$\boldsymbol{m=2.}$}
We change the basis $(Q^1,\dots,Q^n)$ in such a way that $\tau^1$, $\tau^2$
are linearly independent and $\tau^i=0$, $i=3,\dots,n$.
Then $\pi^t_*\mathfrak g_f=\langle \tau^1\p_t,\tau^2\p_t\rangle$ is a faithful
realization of a two-dimensional Lie algebra on the $t$-line,
and according to the Lie theorem, up to linearly combining $Q^1$ and~$Q^2$,
we can set $(\tau^1,\tau^2)=(1,t)$ $(\!{}\bmod \pi^t_*G^\sim)$.
The possible values of $n=\dim\mathfrak g_f$ are $2$, $3$ and~$4$.

Suppose that $n=4$.
Then $\dim \pi^x_*\mathfrak g_f=\dim \langle \xi^3\p_x,\xi^4\p_x\rangle=2$
in view of Lemma~\ref{lem:NonlinKGEqsConditionsForAppropriateSubalgebras}(i).
Up to linearly combining~$Q^3$ and~$Q^4$, we can set $(\xi^3,\xi^4)=(1,x)$ $(\!{}\bmod \pi^x_*G^\sim)$.
Since $\xi^1,\xi^2\in\langle\xi^3,\xi^4\rangle$, we can further linearly combine $Q^1$ and~$Q^2$ with~$Q^3$ and~$Q^4$
to annihilate $\xi^1$ and~$\xi^2$.
We have
\begin{gather*}
Q^1=\p_t+(\eta^{11}u+\eta^{01})\p_u,\quad
Q^2=t\p_t+(\eta^{12}u+\eta^{02})\p_u,\\
Q^3=\p_x+(\eta^{13}u+\eta^{03})\p_u,\quad
Q^4=x\p_x+(\eta^{14}u+\eta^{04})\p_u,
\end{gather*}
where $\eta^{11}$, \dots, $\eta^{14}$ are constants and $\eta^{01}$, \dots, $\eta^{04}$ are smooth functions of~$(t,x)$.
The commutation relations $[Q^1, Q^2]=Q^1$ and $[Q^3, Q^4]=Q^3$ imply $\eta^{11}=\eta^{13}=0$.
Acting by a transformation $\varpi_*\mathscr Z(U^0)$, we can set $\eta^{01}=0$.
From the commutation relation $[Q^1, Q^3]=\eta^{03}_t\p_u=0$, we derive $\eta^{03}_t=0$.
All transformations $\varpi_*\mathscr Z(U^0)$ with $U^0=U^0(x)$ preserve the reduced form of $Q^1=\p_t$,
and among them there is a transformation allowing us to set $\eta^{03}=0$.
We still need to take into account the following commutation relations:
\begin{gather*}
[Q^1, Q^2]=\p_t+\eta^{02}_t\p_u=Q^1, \quad
[Q^3, Q^2]=\eta^{02}_x\p_u=0, \\
[Q^3, Q^4]=\p_x+\eta^{04}_x\p_u=Q^3,\quad
[Q^1, Q^4]=\eta^{04}_t\p_u=0,
\end{gather*}
which leads to $\eta^{02}_t=\eta^{02}_x=\eta^{04}_t=\eta^{04}_x=0$, i.e.,
$\eta^{02}$ and $\eta^{04}$ are constants.
Substituting the components of the vector fields $Q^1$, $Q^2$, $Q^3$ and $Q^4$
into the classifying equation~\eqref{eq:NonlinKGEqsClassifyingEq},
we obtain a system with respect to the arbitrary element~$f$,
\[
f_t=f_x=0,\quad (\eta^{12}u+\eta^{02})f_u=f,\quad (\eta^{14}u+\eta^{04})f_u=f.
\]
This system is consistent if and only if $\eta^{12}=\eta^{14}$ and $\eta^{02}=\eta^{04}$.
Moreover, $\eta^{12}=\eta^{14}\ne0$ since otherwise we have the Liouville equation.
In view of the derived conditions, we can set $\eta^{02}=\eta^{04}=0$ using a shift of~$u$ by a constant,
which preserves all the posed constraints of the components of~$Q^1$, \dots, $Q^4$.
Denoting $p:=-1/\eta^{12}$, we derive Case~\ref{EKGcase12},
where the constant multiplier~$\hat C$ of~$f$ can be removed by the transformation~$\mathscr D^t(\hat C t)$.

Let $n=3$. Hence $\xi^3\ne0$ and $k:=\dim\pi^x_*\mathfrak g_f=\dim\langle\xi^1,\xi^2,\xi^3\rangle\in\{1,2\}$.

If $k=2$, then the Lie theorem implies that $\xi^i=a_ix+b_i$ with some constants~$a_i$ and~$b_i$, $i=1,2,3$.
Therefore, we need to examine the cases with respect to coefficients of $\xi^i$.
For $a_3\ne0$, we can set $a_3=1$ and $b_3=0$
by rescaling the vector field~$Q^3$ and by shifting~$x$, respectively.
We can further linearly combine $Q^1$ and~$Q^2$ with~$Q^3$ to make $a_1=a_2=0$.
Thus, the basis elements~$Q^1$, $Q^2$ and~$Q^3$ satisfy the commutation relations
$[Q^1,Q^2]=Q^1$, \mbox{$[Q^1,Q^3]=0$} and $[Q^2,Q^3]=0$.
The last two commutation relations imply $b_1=b_2=0$,
which means $k=1$,
contradicting the supposed condition $k=2$.
For $a_3=0$, the condition $\xi^3\ne0$ is equivalent to $b_3\ne0$.
Rescaling~$Q^3$ and linearly combining $Q^1$ and~$Q^2$ with~$Q^3$, we make $b_3=1$ and $b_1=b_2=0$.
The commutation relation $[Q^1,Q^2]=Q^1$ leads to $a_1=0$ and $\eta^{11}=0$, and thus $a_2\ne0$ since $k=2$.
Acting by a transformation $\varpi_*\mathscr Z(U^0)$, we can set $\eta^{01}=0$.
The basis elements take the form
\begin{gather*}
Q^1=\p_t,\quad
Q^2=t\p_t+a_2x\p_x+(\eta^{12}u+\eta^{02})\p_u, \quad
Q^3=\p_x+(\eta^{13}u+\eta^{03})\p_u,
\end{gather*}
where $\eta^{1i}$ is a constant and $\eta^{0i}$ is a smooth function of~$(t,x)$, $i=2,3$.
From the commutation relation $[Q^1, Q^3]=\eta^{03}_t\p_u=0$, we derive $\eta^{03}_t=0$.
Since all the pushforwards $\varpi_*\mathscr Z(U^0)$ with $U^0=U^0(x)$ preserve the reduced form of $Q^1=\p_t$,
in view of the equation $\eta^{03}_t=0$ we can assume that $\eta^{03}=0$ up to these pushforwards.
More conditions on the basis elements follow from the commutation relations
\begin{gather*}
[Q^1, Q^2]=\p_t+\eta_t^{02}\p_u=Q^1,\quad
[Q^3, Q^2]= a_2\p_x+(\eta_x^{02}-\eta^{13}\eta^{02})\p_u=a_2Q^3,
\end{gather*}
which are $\eta^{13}=0$, $\eta_t^{02}=0$, $\eta_x^{02}=0$, i.e., $\eta^{02}=\const$.
Substituting the components of the vector fields $Q^1$, $Q^2$ and $Q^3$
into the classifying equation~\eqref{eq:NonlinKGEqsClassifyingEq},
we obtain a system with respect to the arbitrary element~$f$,
\[
f_t=f_x=0,\quad (\eta^{12}u+\eta^{02})f_u=(\eta^{12}-a_2-1)f.
\]
If $\eta^{12}=\eta^{02}=0$, then $a_2=-1$ since $f\ne0$, and we get Case~\ref{EKGcase9}.
For $\eta^{12}=0$ and $\eta^{02}\ne0$, the inequality $f_u\ne0$ implies $a_2\ne-1$,
which leads to the Liouville equation.
If $\eta^{12}\ne0$, then the corresponding values of the arbitrary element~$f$ are, up to the shifts of~$u$,
of the form as in Case~\ref{EKGcase12}, where $n=4$.

Consider the case $k=1$. Then $\xi^3\ne0$, and we can assume $\xi^3=1$, $\xi^1=\xi^2=0$.
In view of the commutation relations $[Q^1,Q^2]=Q^1$ and $[Q^1,Q^3]=0$,
analogously to previous cases, we have $\eta^{11}=0$ and we can set $\eta^{01}=\eta^{03}=0$
acting by the transformation $\mathscr Z(U^0)$.
Further expanding commutation relations,
$[Q^1, Q^2]=\p_t+\eta^{02}_t\p_u=Q^1$,
$[Q^3, Q^2]=(\eta^{02}_x-\eta^{13}\eta^{02})\p_u=0$,
we derive $\eta^{02}_t=0$, $\eta^{02}_x=\eta^{13}\eta^{02}$.
From the classifying equation~\eqref{eq:NonlinKGEqsClassifyingEq},
we obtain the system on the arbitrary element~$f$,
\[
f_t=0,\quad f_x+\eta^{13}uf_u=\eta^{13}f,\quad(\eta^{12}u+\eta^{02})f_u=(\eta^{12}-1)f.
\]
If $\eta^{13}=0$, then $f_t=f_x=0$, and thus $\langle t\p_t-x\p_x\rangle\in\mathfrak g_f$,
which contradicts the condition $\xi^2=0$.
Hence $\eta^{13}\ne0$, and we can set $\eta^{13}=1$ by scaling of $x$, and thus $\eta^{02}=C{\rm e}^x$
for some constant~$C$.
If additionally $\eta^{12}\ne 0$, then we can make $\eta^{02}=0$
acting by the transformation $\mathscr Z(-pC{\rm e}^x)$ with $p:=-1/\eta^{12}\ne0$,
which preserves~$Q^1$ and~$Q^3$.
The general solution of the system for~$f$ with the auxiliary inequality $f_{uu}\ne0$
is $f=\hat C{\rm e}^{-px}|u|^pu$, where $\hat C$ is an arbitrary nonzero constant.
The family of equivalence transformations
$\tilde t=\hat Ct$, $\tilde x=-{\rm e}^{-px}/p$, $\tilde u=u$, $\tilde f={\rm e}^{px}f/\hat C$
maps such values of~$f$ to that from Case~\ref{EKGcase12}, where $n=4$.
Otherwise, $\eta^{12}=0$, and thus $C\ne0$ since $f\ne0$.
Scaling~$u$, we can set $C=1$.
From the above system for~$f$, we derive $f=\hat C{\rm e}^{x-{\rm e}^{-x}u}$.
The equivalence transformation
$\tilde t=-\hat Ct$, $\tilde x={\rm e}^x$, $\tilde u=-u$, $\tilde f={\rm e}^{-x}f/\hat C$
leads to Case~\ref{EKGcase11}.

If $n=2$, then $k\in\{0,1,2\}$.

Suppose that $k=2$.
From the commutation relation $[Q^1, Q^2]=Q^1$, we derive $\eta^{11}=0$ and,
modulo the transformations~$\varpi_*\mathscr D^x(X)$ and~$\varpi_*\mathscr Z(U^0)$,
$\xi^1=1$, $\xi^2=x$ and~$\eta^{01}=0$, and then $\eta^{02}_t+\eta^{02}_x=0$.
Hence $\eta^{02}$ is a function of~$x-t$, and
additionally acting by a transformation~$\varpi_*\mathscr Z(\theta)$,
where $\theta$ is also a function of~$x-t$, we can set $\eta^{02}=0$.
The basis vector fields take the form
$Q^1=\p_t+\p_x$, $Q^2=t\p_t+x\p_x+\eta^{12}u\p_u$.
Substituting the components of~$Q^1$ and~$Q^2$
into the classifying equation~\eqref{eq:NonlinKGEqsClassifyingEq},
we obtain the system
\[
f_t+f_x=0, \quad
tf_t+xf_x+\eta^{12}uf_u=(\eta^{12}-2)f.
\]
This system implies that $\eta^{12}\ne0$ since otherwise we get Case~\ref{EKGcase10} with $n=3$.
The condition $\eta^{12}\ne0$ singles out Case~\ref{EKGcase7}, where $q:=-\eta^{12}$.

In the case $k=1$, the commutation relation $[Q^1, Q^2]=Q^1$ implies $\xi^1=0$ and $\eta^{11}=0$.
It is possible and convenient to make $\xi^2=x$ using a transformation~$\varpi_*\mathscr D^x(X)$.
Similarly to the previous cases, we can simultaneously set $\eta^{01}=0$ and $\eta^{02}=0$
via acting by a transformation $\varpi_*\mathscr Z(U^0)$.
The system $f_t=0$, $xf_x+\eta^{12}uf_u=(\eta^{12}-2)f$
following from the classifying equation~\eqref{eq:NonlinKGEqsClassifyingEq}
implies that $\eta^{12}\ne0$ since otherwise we again get Case~\ref{EKGcase9} with $n=3$.
As a result, we have Case~\ref{EKGcase8}, where $q:=-\eta^{12}$.

\begin{remark}\label{rem:NonlinKGEqsGaugingPInCases7And8}
The (nonzero) parameter~$q$ in Cases~\ref{EKGcase7} and~\ref{EKGcase8}
cannot be further gauged by equivalence transformations.
We show this only for Case~\ref{EKGcase8} since the argumentation for Case~\ref{EKGcase7} is similar.
For each value of~$f$ from Case~\ref{EKGcase8},
the derived algebra of the corresponding algebra~$\mathfrak g_f$ is spanned by the vector field~$Q^1=\p_t$.
Hence the projection~$\varpi_*\mathscr T$ of any element~$\mathscr T$ of~$G^\sim$
that maps the equation~$K_f$ to an equation from the same Case~\ref{EKGcase8}
should preserve the span of~$Q^1$, i.e., $(\varpi_*\mathscr T)_*\p_t\in\langle\p_t\rangle$.
This implies that the $t$-component of~$\mathscr T$ is affine in~$t$.
Such a transformation cannot change the ratio of the coefficients of~$t\p_t$ and of~$u\p_u$
in the vector field~$Q^2$.
\end{remark}

For $k=0$, $\xi^1=\xi^2=0$, and the further consideration is similar to the case $k=1$.
Modulo the $G^\sim$-equivalence, we derive from the commutation relation $[Q^1,Q^2]=Q^1$
that $\eta^{11}=0$ and $\eta^{01}=\eta^{02}=0$.
The associated system for~$f$ is $f_t=0$, $\eta^{12}uf_u=(\eta^{12}-1)f$,
where $\eta^{12}\ne0,1$,
and hence such a value of~$f$ can be reduced by a transformation~$\mathscr D^x(X)$
to the form from Case~\ref{EKGcase12}, where $n=4$.

\medskip\par\noindent{$\boldsymbol{m=1.}$}
Then $n\leqslant2$.

If $n=2$, then linearly combining the basis elements and pushing forward the algebra~$\mathfrak g_f$
by $\varpi_*\mathscr D^t(T)$ and $\varpi_*\mathscr Z(U^0)$, we make $\tau^1=1$, $\tau^2=0$ and $\eta^{01}=0$.
Hence $\xi^2\ne0$, i.e., $\xi^2=1$ $(\!{}\bmod G^\sim)$,
and we additionally set $\xi^1=0$ by linearly combining basis elements
and repairing the gauge $\eta^{01}=0$ using a transformation $\varpi_*\mathscr Z(U^0)$.
In view of the commutation relation
$[Q^1, Q^2]=(\eta^{02}_t-\eta^{11}\eta^{02})\p_u=0$,
we can conclude that $\eta^{02}={\rm e}^{\eta^{11}t}\zeta(x)$ for some smooth function~$\zeta$ of~$x$.
Therefore, we can set $\eta^{02}=0$ by the transformation $\varpi_*\mathscr Z\big({\rm e}^{\eta^{11}t}\theta(x)\big)$,
where $\theta$ is a solution of the first-order linear ordinary differential equation $\theta_x-\eta^{12}\theta+\zeta=0$;
the vector field~$Q^1$ is preserved by this transformation.
The basis elements take the form $Q^1=\p_t+\eta^{11}u\p_u$, $Q^2=\p_x+\eta^{12}u\p_u$.
Here $(\eta^{11},\eta^{12})\ne(0,0)$ since otherwise we have Case~\ref{EKGcase9} with $n=3$.
If both the coefficients~$\eta^{11}$ and~$\eta^{12}$ are nonzero,
then by scaling~$t$ and~$x$, we can set $\eta^{11}=\eta^{12}=1$,
and the associated system for the arbitrary element~$f$ is $f_t+uf_u=f$, $f_x+uf_u=f$,
which leads to Case~\ref{EKGcase6}.
If one of the parameters $\eta^{11}$ or $\eta^{12}$ is nonzero,
then up to the discrete equivalence transformation~$\mathscr I^0$,
which permutes the variables~$t$ and~$x$, we can assume that $\eta^{11}\ne0$,
and we can set $\eta^{11}=1$ by scaling~$t$.
From the classifying equation~\eqref{eq:NonlinKGEqsClassifyingEq},
we obtain the system $f_t+uf_u=f$, $f_x=0$, whose integration gives Case~\ref{EKGcase5}.

In the case $n=1$, since $\tau^1\ne0$, we can set $\tau^1=1$, $\xi^1=\delta$, $\eta^{01}=0$ and $\eta^{11}=(1+\delta)\delta'$ with $\delta,\delta'\in \{0,1\}$
modulo the transformations~$\varpi_*\mathscr D^t(T)$, $\varpi_*\mathscr D^x(X)$ and~$\varpi_*\mathscr Z(U^0)$
and of a simultaneous scaling of~$(t,x)$, respectively.
As a result, we obtain the following $G^\sim$-inequivalent cases for $Q^1$:
$\p_t$, $\p_t+\p_x$, $\p_t+u\p_u$, and $\p_t+\p_x+2u\p_u$,
which correspond to Cases~\ref{EKGcase1}--\ref{EKGcase4}.

\begin{remark}\label{rem:NonlinKGEqsRealizations}
We have split the group classification of equations from the class~\eqref{eq:NonlinKGEqs}
with finite-dimensional maximal Lie invariance algebras
into different cases depending on values of the triple $(m,n,k)$
of $G^\sim$-invariant integers, which are defined by~\eqref{eq:NonlinKGEqsMNK}.
It is clear that most of the values in~$\mathbb Z^3$ are inappropriate for $(m,n,k)$.
As a preliminary step of the classification,
we have significantly narrowed down the set of candidates
to be looked through for such triples.
According to Lemma~\ref{lem:NonlinKGEqsConditionsForAppropriateSubalgebras},
Corollary~\ref{cor:NonlinKGEqsProjectionDim} and the definition of $(m,n,k)$,
we have the constraints
$m\in\{1,2,3\}$, $n\in\{1,2,3,4\}$, and $0\leqslant n-m\leqslant k\leqslant m\leqslant n$
but even the (quite restricted) set~$S$ of triples satisfying these constraints
contains many elements that are not realized for equations from the class~\eqref{eq:NonlinKGEqs}.
The appropriate values for $(m,n,k)$ are exhausted by
\[
(1,1,0),\ \ 
(1,1,1),\ \ 
(1,2,1),\ \ 
(2,2,2),\ \ 
(2,2,1),\ \ 
(2,3,2),\ \ 
(3,3,3),\ \ 
(2,3,1),\ \ 
(2,4,2),
\]
which are associated with pairs of Cases~\ref{EKGcase1} and~\ref{EKGcase3},
\ref{EKGcase2} and~\ref{EKGcase4},
\ref{EKGcase5} and~\ref{EKGcase6},
and single Cases~\ref{EKGcase7}, \ref{EKGcase8}, \ref{EKGcase9}, \ref{EKGcase10}, \ref{EKGcase11}, \ref{EKGcase12}
of Theorem~\ref{thm:NonlinKGEqsGroupClassification}, respectively.
\noprint{
$(m,n,k)=(1,1,0)$ that gives Cases~\ref{EKGcase1}, \ref{EKGcase3};
$(m,n,k)=(1,1,1)$ -- Cases~\ref{EKGcase2}, \ref{EKGcase4};
$(m,n,k)=(1,2,1)$ -- Cases~\ref{EKGcase5}, \ref{EKGcase6};
$(m,n,k)=(2,2,2)$ -- Case~\ref{EKGcase7};
$(m,n,k)=(2,2,1)$ -- Case~\ref{EKGcase8};
$(m,n,k)=(2,3,2)$ -- Case~\ref{EKGcase9}.
$(m,n,k)=(3,3,3)$ -- Case~\ref{EKGcase10}.
$(m,n,k)=(2,3,1)$ -- Case~\ref{EKGcase11}.
$(m,n,k)=(2,4,2)$ -- Case~\ref{EKGcase12}.
}
Therefore, the inappropriate triples in~$S$ are
$(3,4,k)$ with $k=1,2,3$, $(3,3,k)$ with $k=0,1,2$ and $(2,2,0)$,
and their number is significant but less than the number of appropriate triples.
The only appropriate triple with $m=3$ is $(3,3,3)$,
i.e., the value $m=3$ uniquely defines the possible values for~$n$ and~$k$.
This interesting observation may be related to the fact that for $m=3$
both the projections~$\pi^t_*\mathfrak g_f$ and~$\pi^x_*\mathfrak g_f$ as well as the algebra~$\mathfrak g_f$ itself
are necessarily isomorphic to $\mathfrak{sl}(2,\mathbb R)$
in view of the Lie theorem and the simplicity of $\mathfrak{sl}(2,\mathbb R)$.
The separation of the appropriate values from the inappropriate ones in the set~$S$
cannot be implemented in the course of preliminary analysis of Lie symmetries of equations from the class~\eqref{eq:NonlinKGEqs}
since it is an integral part of the group classification procedure for this class.
\end{remark}

\begin{remark}\label{rem:IdentifyingInvValues}
To check algebraically that the cases given in Theorem~\ref{thm:NonlinKGEqsGroupClassification} are $G^\sim$-inequivalent to each other,
we need more values associated with the maximal Lie invariance algebras
of equations from the class~\eqref{eq:NonlinKGEqs},
since there are pairs of $G^\sim$-inequivalent cases with the same triples $(m,n,k)$.
These are
Cases~\ref{EKGcase1} and~\ref{EKGcase3} with $(m,n,k)=(1,1,0)$,
Cases~\ref{EKGcase2} and~\ref{EKGcase4} with $(m,n,k)=(1,1,1)$ and
Cases~\ref{EKGcase5} and~\ref{EKGcase6} with $(m,n,k)=(1,2,1)$.
To introduce additional $G^\sim$-invariant values for complete identification of classification cases,
we represent the span $\mathfrak g_\spanindex$ as a direct sum of its subspaces,
\[
\mathfrak g_\spanindex=\big\langle D^t(\tau)\big\rangle\dotplus\big\langle D^x(\xi)\big\rangle\dotplus\big\langle I\big\rangle\dotplus\big\langle Z(\eta^0)\big\rangle,
\]
where the parameter functions $\tau=\tau(t)$, $\xi=\xi(x)$ and~$\eta^0=\eta^0(t,x)$ run through the sets of smooth functions of their arguments.
By $\mathfrak P_i$ we denote the projection from~$\mathfrak g_\spanindex$
onto the $i$th summand of the above representation for~$\mathfrak g_\spanindex$, $i=1,\dots,4$.
Although
$\dim\mathfrak P_1\mathfrak s=\dim\pi^t_*\mathfrak s$ and
$\dim\mathfrak P_2\mathfrak s=\dim\pi^x_*\mathfrak s$ for subalgebras~$\mathfrak s$ of~$\mathfrak g_\spanindex$,
we can define the new $G^\sim$-invariant integer values
\begin{gather*}
l=l(\mathfrak s):=\dim\mathfrak P_3\mathfrak s,\\
j_1=j_1(\mathfrak s):=\max(\dim\mathfrak s^1,\dim\mathfrak s^2),\quad
j_2=j_2(\mathfrak s):=\min(\dim\mathfrak s^1,\dim\mathfrak s^2),\\
j_{12}=j_{12}(\mathfrak s):=\dim\mathfrak s^{12},\\
j_{13}=j_{13}(\mathfrak s):=\max(\dim\mathfrak s^{13},\dim\mathfrak s^{23}),\quad
j_{23}=j_{23}(\mathfrak s):=\min(\dim\mathfrak s^{13},\dim\mathfrak s^{23}),\\
r_1=r_1(\mathfrak s):=\max(\dim\pi^t_*\mathfrak s^{12},\dim\pi^x_*\mathfrak s^{12}),\quad
r_2=r_2(\mathfrak s):=\min(\dim\pi^t_*\mathfrak s^{12},\dim\pi^x_*\mathfrak s^{12}),\\
r_3=r_3(\mathfrak s):=3-\min\big\{\dim\big\langle D^t(\tau),D^x(\xi )\big\rangle\mid \exists\,\eta^0\colon D^t(\tau)+D^x(\xi )+I+Z(\eta^0)\in\mathfrak s\big\},
\end{gather*}
where $r_3:=0$ if the set in the definition of~$r_3$ is empty,
\begin{gather*}
\mathfrak s^1:=\mathfrak s\cap\big\langle D^t(\tau),Z(\eta^0)\big\rangle,\quad
\mathfrak s^2:=\mathfrak s\cap\big\langle D^x(\xi ),Z(\eta^0)\big\rangle,\\
\mathfrak s^{12}:=\mathfrak s\cap\big\langle D^t(\tau),D^x(\xi ),Z(\eta^0)\big\rangle,\\
\mathfrak s^{13}:=\mathfrak s\cap\big\langle D^t(\tau),I,Z(\eta^0)\big\rangle,\quad
\mathfrak s^{23}:=\mathfrak s\cap\big\langle D^x(\xi ),I,Z(\eta^0)\big\rangle.
\end{gather*}
The values of $(m,n,k,l,j_1,j_2,j_{12},j_{13},j_{23},r_1,r_2,r_3)$ for $\mathfrak s=\mathfrak g_f$
differ from each other for different cases of Theorem~\ref{thm:NonlinKGEqsGroupClassification},
\begin{gather*}
\phantom{1}\ref{EKGcase1} .\ (1,1,0,0,1,0,1,1,0,1,0,0);\quad
\phantom{1}\ref{EKGcase2} .\ (1,1,1,0,0,0,1,0,0,1,1,0);\\
\phantom{1}\ref{EKGcase3} .\ (1,1,0,1,0,0,0,1,0,0,0,2);\quad
\phantom{1}\ref{EKGcase4} .\ (1,1,1,1,0,0,0,0,0,0,0,1);\\
\phantom{1}\ref{EKGcase5} .\ (1,2,1,1,1,0,1,1,1,1,0,2);\quad
\phantom{1}\ref{EKGcase6} .\ (1,2,1,1,0,0,1,1,1,1,1,2);\\
\phantom{1}\ref{EKGcase7} .\ (2,2,2,1,0,0,1,0,0,1,1,1);\quad
\phantom{1}\ref{EKGcase8} .\ (2,2,1,1,1,0,1,1,0,1,0,1);\\
\phantom{1}\ref{EKGcase9} .\ (2,3,2,0,1,1,3,1,1,2,2,0);\quad
           \ref{EKGcase10}.\ (3,3,3,0,0,0,3,0,0,3,3,0);\\
           \ref{EKGcase11}.\ (2,3,1,1,2,0,2,2,1,2,0,2);\quad
           \ref{EKGcase12}.\ (2,4,2,1,1,1,3,2,1,2,2,2);\\
           \ref{EKGcase13}.\ (\infty,\infty,\infty,0,\infty,\infty,\infty,\infty,\infty,\infty,\infty,0).
\end{gather*}
This is why these cases are $G^\sim$-inequivalent.
At the same time, this integer tuple is redundant for distinguishing the classification cases from each other.
The most remarkable minimal sufficient tuple is the triple $(r_3,j_1,r_2)$,
where we order the characteristics according to their importance.
The other sufficient triples are
$(r_3,j_1,n)$, $(r_3,j_1,k)$, $(r_3,j_1,j_{12})$, $(r_3,j_1,r_1)$, $(r_3,r_2,n)$, $(r_3,r_2,j_{12})$, $(r_3,r_2,r_1)$.
Nevertheless, in the course of the study of successive extension in the next section,
we need to extend triples with other values among the above ones,
although the values~$r_3$, $j_1$ and~$r_2$ jointly with~$n$ are still of primary importance.
\end{remark}

\section{Successive Lie-symmetry extensions}\label{sec:SuccessiveLie-SymExtensions} 

Throughout this section, the classification cases listed in Theorem~\ref{thm:NonlinKGEqsGroupClassification}
are interpreted in the weak sense.
We intend to identify all the pairs $(\mbox{Case\ }N,\mbox{Case\ }\bar N)$ of $G^\sim$-equivalent Lie-symmetry extensions
with $\mbox{Case\ }N\prec\mbox{Case\ }\bar N$,
i.e., where Case~$\bar N$ is an additional Lie-symmetry extension of Case~$N$ modulo the $G^\sim$-equivalence,
see Remark~\ref{rem:NonlinKGEqsOnOrderingSuccessiveLieSymExtensions}.
For this purpose, we use a technique similar to that for the classification of contractions
of low-dimensional Lie algebras, see, e.g., \cite{burd1999a,grun1988a,nest2006a} and references therein.
Let $\mathfrak s$ and~$\bar{\mathfrak s}$ be subalgebras of~$\mathfrak g_\spanindex$
associated with Cases~$N$ and~$\bar N$, and
\looseness=-1
\[
(m,n,k,l,j_1,j_2,j_{12},j_{13},j_{23},r_1,r_2,r_3)\quad\mbox{and}\quad
(\bar m,\bar n,\bar k,\bar l,\bar j_1,\bar j_2,\bar j_{12},\bar j_{13},\bar j_{23},\bar r_1,\bar r_2,\bar r_3)
\]
are the tuples of their $G^\sim$-invariant characteristics that are defined
in Remarks~\ref{rem:NonlinKGEqsRealizations} and~\ref{rem:IdentifyingInvValues}.
It is obvious that the relation $\mbox{Case\ }N\prec\mbox{Case\ }\bar N$ implies
\begin{gather*}
n<\bar n,\quad
m\leqslant\bar m,\quad
k\leqslant\bar k,\quad
l\leqslant\bar l,\quad
j_1\leqslant\bar j_1,\quad
j_2\leqslant\bar j_2,\\
j_{12}\leqslant\bar j_{12},\quad
j_{13}\leqslant\bar j_{13},\quad
j_{23}\leqslant\bar j_{23},\quad
r_1\leqslant\bar r_1,\quad
r_2\leqslant\bar r_2,\quad
r_3\leqslant\bar r_3.
\end{gather*}
In other words, if at least one of the above inequalities does not hold,
then $\mbox{Case\ }N\nprec\mbox{Case\ }\bar N$.
Examining all the pairs of the cases listed in Theorem~\ref{thm:NonlinKGEqsGroupClassification},
we exclude the pairs $(\mbox{Case\ }N,\mbox{Case\ }\bar N)$ with $\mbox{Case\ }N\nprec\mbox{Case\ }\bar N$.
It turns out that for this exclusion it suffices to use only a tuple of five of the above $G^\sim$-invariant integer values,
e.g., $(n,r_3,r_2,j_1,k)$, which is minimally sufficient.
The other minimally sufficient tuples of five characteristics are obtained by replacing~$k$ by $m$ or~$r_1$.
We order the characteristics according to their importance in the elimination procedure.
The principal characteristic is the dimension~$n$ of the entire general Lie invariance algebra of
the corresponding case,
and the inequality between~$n$ and~$\bar n$ should only be strict
for $\mbox{Case\ }N$ and $\mbox{Case\ }\bar N$ to be ordered.
The characteristics $r_3$, $r_2$, $j_1$, $k$, $m$ and~$r_1$ detect the following cases of disordering with~$n<\bar n$:
\begin{align*}
r_3\colon\ &
\mbox{Case\ }\ref{EKGcase3}\nprec\mbox{Cases\ }\ref{EKGcase7},\ref{EKGcase8},\ref{EKGcase9},\ref{EKGcase10},\ref{EKGcase13},\quad
\mbox{Cases\ }\ref{EKGcase4},\ref{EKGcase5},\ref{EKGcase6},\ref{EKGcase7},\ref{EKGcase8}\nprec\mbox{Cases\ }\ref{EKGcase9},\ref{EKGcase10},\ref{EKGcase13},\\
&\mbox{Cases\ }\ref{EKGcase11},\ref{EKGcase12}\nprec\mbox{Case\ }\ref{EKGcase13},\\
r_2\colon\ &
\mbox{Case\ }\ref{EKGcase2}\nprec\mbox{Cases\ }\ref{EKGcase5},\ref{EKGcase8},\ref{EKGcase11},\quad
\mbox{Cases\ }\ref{EKGcase6},\ref{EKGcase7}\nprec\mbox{Case\ }\ref{EKGcase11},\\
j_1\colon\ &
\mbox{Case\ }\ref{EKGcase1}\nprec\mbox{Cases\ }\ref{EKGcase6},\ref{EKGcase7},\ref{EKGcase10},\quad
\mbox{Cases\ }\ref{EKGcase5},\ref{EKGcase8}\nprec\mbox{Case\ }\ref{EKGcase10},\quad
\mbox{Case\ }\ref{EKGcase11}\nprec\mbox{Case\ }\ref{EKGcase12},\\
k\colon\ &
\mbox{Case\ }\ref{EKGcase7}\nprec\mbox{Case~\ref{EKGcase11}},\quad
\mbox{Case\ }\ref{EKGcase10}\nprec\mbox{Case\ }\ref{EKGcase12},\quad
\mbox{or}\quad m,r_1\colon\ \mbox{Case\ }\ref{EKGcase10}\nprec\mbox{Case\ }\ref{EKGcase12}.
\end{align*}
The direct inspection shows that the remaining pairs $(\mbox{Case\ }N,\mbox{Case\ }\bar N)$ with~$n<\bar n$
are necessarily ordered, except the pairs
$(\mbox{Case~\ref{EKGcase7}},\mbox{Case~\ref{EKGcase10}})$ and
$(\mbox{Case~\ref{EKGcase8}},\mbox{Case~\ref{EKGcase9}})$,
which are related to limit processes for Cases~\ref{EKGcase7} and~\ref{EKGcase8} as $q\to0$.
Therefore, the Hasse diagram in Figure~\ref{fig:HasseDiagram} represents
the structure of the partially ordered set of Lie-symmetry extensions
within the class~\eqref{eq:NonlinKGEqs}, cf.\ Remark~\ref{rem:NonlinKGEqsOnHasseDiagram}.
Note that there are two characteristics, $j_2$ and~$j_{12}$, that detect no cases of disordering with~$n<\bar n$.
Each of the characteristics $l$, $j_{13}$ and $j_{23}$ detects only cases of disordering with~$n<\bar n$,
that are detected by other characteristics.
For example, the characteristic~$l$ detects
$\mbox{Cases\ }\ref{EKGcase3},\ref{EKGcase4},\ref{EKGcase5},\ref{EKGcase6},\ref{EKGcase7},\ref{EKGcase8}\nprec\mbox{Cases\ }\ref{EKGcase9},\ref{EKGcase10},\ref{EKGcase13}$ and
$\mbox{Cases\ }\ref{EKGcase11},\ref{EKGcase12}\nprec\mbox{Case\ }\ref{EKGcase13}$,
which is completely covered by the characteristic~$r_3$.

We derive the necessary and sufficient conditions for the parameter function~$\hat f$
under which equations from Cases~\ref{EKGcase1}--\ref{EKGcase9} have wider Lie invariance algebras
than equations with general values of~$\hat f$.
Here we omit Case~\ref{EKGcase10}
since we have shown in Section~\ref{sec:NonlinKGEqsProof} that this case admits no further Lie-symmetry extension.
We first consider Cases~\ref{EKGcase1}--\ref{EKGcase4},
for each of which the parameter function~$\hat f$ depends on two arguments
and the corresponding common Lie invariance algebra is one-dimensional.
By default, we assume that
the second derivative of $\hat f$ with respect to the argument involving~$u$ is nonzero.

Case~\ref{EKGcase1} possesses, modulo the $G^\sim$-equivalence, three families of further Lie-symmetry extensions,
which are given by Cases~\ref{EKGcase5}, \ref{EKGcase8} and~\ref{EKGcase9}.
Analyzing them, we conclude that for any further Lie-symmetry extension of Case~\ref{EKGcase1},
the corresponding invariance algebra contains $Q^2\in \mathfrak g_\spanindex$
with $\pi^x_*Q^2\ne0$ and $[Q^1,Q^2]\in\langle Q^1,Q^2\rangle$.
Therefore,  $[Q^1,Q^2]\in\langle Q^1\rangle$.
Up to rescaling of~$Q^2$, we can assume $[Q^1,Q^2]=\delta Q^1$ with $\delta\in\{0,1\}$.
We split the last commutation relation componentwise and integrate the obtained equations
for the components of~$Q^2$.
Linearly combining~$Q^2$ with~$Q^1$ if necessary, we derive the representation
$Q^2=\delta t\p_t +\xi(x)\p_x+\big(\eta^1u+\eta^0(x)\big)\p_u$,
where $\xi$ and~$\eta^0$ are arbitrary smooth functions of~$x$ with~$\xi\ne0$,
and $\eta^1$ is an arbitrary constant.
The substitution of this representation into the classifying equation~\eqref{eq:NonlinKGEqsClassifyingEq}
leads to the equation
\begin{gather}\label{eq:ExtCase1}
\xi\hat f_x +\big(\eta^1u+\eta^0\big)\hat f_u=\big(\eta^1-\delta-\xi_x\big)\hat f.
\end{gather}
For any value of the parameter function~$\hat f$ satisfying the last equation,
we indeed have a further Lie-symmetry extension of Case~\ref{EKGcase1},
which belongs, up to the $G^\sim$-equivalence,
to Case~\ref{EKGcase5} if $\eta^1\ne0$ and $\delta=0$,
to Case~\ref{EKGcase8} if $\eta^1\ne0$ and $\delta=1$, or
to Case~\ref{EKGcase9} if $\eta^1=0$.

Case~\ref{EKGcase3} is considered similarly to Case~\ref{EKGcase1}.
The further Lie-symmetry extensions of Case~\ref{EKGcase3} are exhausted, modulo the $G^\sim$-equivalence,
by Cases~\ref{EKGcase5} and~\ref{EKGcase6}.
The additional Lie-symmetry vector field $Q^2\in \mathfrak g_\spanindex$
satisfies the conditions $\pi^x_*Q^2\ne0$ and $[Q^1,Q^2]=0$.
This is why we can assume without loss of generality that,
up to linearly combining~$Q^2$ with~$Q^1$ and rescaling~$Q^2$,
$Q^2=\xi(x)\p_x+(\delta u+\eta^0(x))\p_u$,
where $\xi$ and~$\eta^0$ are arbitrary smooth functions of~$x$ with~$\xi\ne0$,
and $\delta\in\{0,1\}$.
Substituting such $Q^2$ into the classifying equation~\eqref{eq:NonlinKGEqsClassifyingEq}
and successively splitting with respect to~$t$ under assuming $x$ and~$\omega:={\rm e}^{-t}u$ as the other independent variables,
we derive one more constraint $\eta^0=0$ for components of~$Q^2$
and the equation
\begin{gather}\label{eq:ExtCase3}
\xi\hat f_x+\delta\omega\hat f_\omega=\big(\delta-\xi_x\big)\hat f.
\end{gather}
The last equation defines, up to the $G^\sim$-equivalence, further Lie-symmetry extensions
to Case~\ref{EKGcase5} or Case~\ref{EKGcase6} if $\delta=0$ or $\delta=1$, respectively.

Up to the $G^\sim$-equivalence, Case~\ref{EKGcase2} has further Lie-symmetry extensions to Cases~\ref{EKGcase6},~\ref{EKGcase7}, \ref{EKGcase9} and~\ref{EKGcase10}.
For the additional Lie-symmetry vector field $Q^2\in\mathfrak g_\spanindex$,
we have $[Q^1,Q^2]=\delta Q^1+\kappa Q^2$ for some constants~$\delta$ and~$\kappa$.
If $\kappa=0$, then, up to rescaling of~$Q^2$ and linearly recombining~$Q^2$ with~$Q^1$,
we can set $Q^2=(\delta t+\kappa')\p_t+\delta x\p_x+\big(\eta^1u+\hat\eta^0(\omega)\big)\p_u$,
where $\hat\eta^0$ is an arbitrary smooth function of~$\omega:=x-t$,
$\eta^1$ is an arbitrary constant,
$\delta\in\{0,1\}$, $\kappa'$ is an arbitrary constant if $\delta=1$, and $\kappa'=1$ if $\delta=0$.
Analogously to the previous cases, we substitute $Q^2$ into the classifying equation~\eqref{eq:NonlinKGEqsClassifyingEq}
and obtain the equation
\begin{gather}\label{eq:ExtCase2a}
(\delta\omega-\kappa')\hat f_\omega + \big(\eta^1u+\hat\eta^0(\omega)\big)\hat f_u=\big(\eta^1-2\delta\big)\hat f-\hat\eta^0_{\omega\omega}.
\end{gather}
In view of this equation, up to the $G^\sim$-equivalence,
we have further Lie-symmetry extensions
to Case~\ref{EKGcase6}  if $\delta=0$ and $\eta^1\ne0$,
to Case~\ref{EKGcase7}  if $\delta=1$ and $\eta^1\ne0$,
to Case~\ref{EKGcase9}  if $\delta=\eta^1=0$, and
to Case~\ref{EKGcase10} if $\delta=1$ and $\eta^1=0$.
If $\kappa\ne0$, then $\eta^1=0$. 
Linearly recombining~$Q^2$ with~$Q^1$, we can set $\delta=0$.
Hence $Q^2=C_1{\rm e}^{\kappa t}\p_t+C_2{\rm e}^{\kappa x}\p_x+{\rm e}^{\kappa t}\hat\eta^0(\omega)\p_u$,
where $\hat\eta^0$ is again an arbitrary smooth function of~$\omega:=x-t$,
and $C_1$ and~$C_2$ are arbitrary constants with $(C_1,C_2)\ne(0,0)$.
The classifying equation~\eqref{eq:NonlinKGEqsClassifyingEq} with such~$Q^2$ results in the equation
\begin{gather}\label{eq:ExtCase2b}
(C_2{\rm e}^{\kappa\omega} -C_1)\hat f_\omega+\hat\eta^0\hat f_u=-\kappa(C_1+C_2{\rm e}^{\kappa\omega})\hat f+\kappa\hat\eta^0_\omega-\hat\eta^0_{\omega\omega}.
\end{gather}
Here the conditions $C_1C_2=0$ and $C_1C_2\ne 0$ are associated with further Lie-symmetry extensions
to Cases~\ref{EKGcase9} and~\ref{EKGcase10}, respectively.

\looseness=1
All the classification cases with $n>1$ and $l>0$ are, up to the $G^\sim$-equivalence,
further Lie-symmetry extensions of Case~\ref{EKGcase4}.
Its direct Lie-symmetry extensions are exhausted, modulo the $G^\sim$-equivalence,
by Cases~\ref{EKGcase5}, \ref{EKGcase6}, \ref{EKGcase7} and~\ref{EKGcase8}.
In view of the form of~$Q^1$, we have the following commutation relation of~$Q^1$ with
the additional Lie-symmetry vector field $Q^2\in\mathfrak g_\spanindex$:
$[Q^1,Q^2]=\kappa Q^2$ for some constant $\kappa$.
The commutation relation implies the representation
$Q^2=C_1{\rm e}^{\kappa t}\p_t+C_2{\rm e}^{\kappa x}\p_x+{\rm e}^{(\kappa+2)t}\hat\eta^0(\omega_1)\p_u$,
where $\hat\eta^0$ is again an arbitrary smooth function of~$\omega_1:=x-t$,
and $C_1$ and~$C_2$ are arbitrary constants with \mbox{$(C_1,C_2)\ne(0,0)$} and, if $\kappa=0$, additionally $C_1\ne C_2$.
Substituting this representation into the classifying equation~\eqref{eq:NonlinKGEqsClassifyingEq},
we obtain the equation
\begin{gather}\label{eq:ExtCase4}
\begin{split}
&\big(C_2{\rm e}^{\kappa\omega_1}-C_1\big)\hat f_{\omega_1}+
\big({\rm e}^{-\omega_1}\hat\eta^0(\omega_1)- C_2\omega_2{\rm e}^{\kappa\omega_1}-C_1\omega_2\big)\hat f_{\omega_2}\\
&\qquad =
-(\kappa+1) \big(C_1 +C_2{\rm e}^{\kappa\omega_1}\big)\hat f+
(\kappa+2){\rm e}^{-\omega_1} \hat\eta^0_{\omega_1}- {\rm e}^{-\omega_1} \hat\eta^0_{\omega_1\omega_1},
\end{split}
\end{gather}
where $\omega_2:={\rm e}^{-x-t}u$.
Modulo the $G^\sim$-equivalence, we obtain extensions
to Case~\ref{EKGcase5} if $\kappa=0$   and $C_1C_2=0$,
to Case~\ref{EKGcase6} if $\kappa=0$   and $C_1C_2\ne0$,
to Case~\ref{EKGcase7} if $\kappa\ne0$ and $C_1C_2\ne0$, and
to Case~\ref{EKGcase8} if $\kappa\ne0$ and $C_1C_2=0$.

We summarize the above consideration in the following proposition.

\begin{proposition}\label{pro:FurtherLieSymExtestionsForCases1-4}
A generalized nonlinear Klein--Gordon equation from Cases~\ref{EKGcase1}--\ref{EKGcase4} admits an additional Lie-symmetry extension
if and only if the corresponding value of the parameter function~$\hat f$ satisfies
an equation~\eqref{eq:ExtCase1} in Case~\ref{EKGcase1},
an equation~\eqref{eq:ExtCase3} in Case~\ref{EKGcase3},
an equation~\eqref{eq:ExtCase2a} or~\eqref{eq:ExtCase2b} in Case~\ref{EKGcase2}, or
an equation~\eqref{eq:ExtCase4} in Case~\ref{EKGcase4}.
\end{proposition}

Now we derive the conditions on the parameter function~$\hat f$
for the Lie invariance algebras presented in Cases~\ref{EKGcase5}--\ref{EKGcase9}
of Theorem~\ref{thm:NonlinKGEqsGroupClassification}
to be maximal for the corresponding equations from the class~\eqref{eq:NonlinKGEqs}.
In each of these cases, the arbitrary element~$f$ takes a value of the form
$f=\alpha(t,x)\hat f(\omega)$, where $\omega:=\beta(t,x)u$,
$\alpha$ and $\beta$ are nonzero known functions of~$(t,x)$,
and $\hat f_{\omega\omega}\ne0$ since $f_{uu}\ne0$.
Substituting this form for~$f$ into the classifying equation~\eqref{eq:NonlinKGEqsClassifyingEq},
we obtain the classifying equation in terms of~$\hat f$,
\begin{gather}\label{eq:ClassifyingEqForCases5-10}
\bigg(\bigg(\frac{\beta_t}\beta\tau+\frac{\beta_x}\beta\xi+\eta^1\bigg)\omega+\beta\eta^0\bigg)\hat f_\omega
+\bigg(\tau_t+\xi_x+\frac{\alpha_t}\alpha\tau+\frac{\alpha_x}\alpha\xi-\eta^1\bigg)\hat f-\frac{\eta^0_{tx}}\alpha=0.
\end{gather}
We apply the method of furcate splitting, see~\cite{bihl2020a,niki2001a,opan2020b} and references therein.
Fixing values of the variables~$t$ and~$x$ gives the template form of equations for values of~$\hat f$,
for which the equation~$K_f$ possesses an additional Lie-symmetry extension,
\begin{gather}\label{eq:TemplateFormForCases5-10}
(a\omega+b)\hat f_\omega+c\hat f-d=0,
\end{gather}
where $a$, $b$, $c$ and~$d$ are constants with $(a,b)\ne(0,0)$.
Additionally, in view of $\hat f_{\omega\omega}\ne0$ we have $c\ne-a$ if $a\ne0$ and $c\ne0$ if $a=0$.
Moreover, the number of equations of the form~\eqref{eq:TemplateFormForCases5-10}
with linearly independent tuples $(a,b,c,d)$ cannot exceed one since otherwise $\hat f_{\omega\omega}=0$.
In other words, we have exactly one independent equation of the form~\eqref{eq:TemplateFormForCases5-10}
if the equation~$K_f$ possesses an additional Lie-symmetry extension.
This means that the left-hand side of~\eqref{eq:ClassifyingEqForCases5-10} is proportional to
that of~\eqref{eq:TemplateFormForCases5-10} with nonvanishing multiplier~$\lambda$ depending on~$(t,x)$,
\begin{gather*}
\bigg(\bigg(\frac{\beta_t}\beta\tau+\frac{\beta_x}\beta\xi+\eta^1\bigg)\omega+\beta\eta^0\bigg)\hat f_\omega
+\bigg(\tau_t+\xi_x+\frac{\alpha_t}\alpha\tau+\frac{\alpha_x}\alpha\xi-\eta^1\bigg)\hat f-\frac{\eta^0_{tx}}\alpha\\
\qquad{}=\lambda\big((a\omega+b)\hat f_\omega+c\hat f-d\big).
\end{gather*}
The last equation can be split with respect to $\hat f$ and~$\hat f_\omega$ into the system
\begin{gather}\label{eq:SplitClassifyingEqForCases5-10}
\frac{\beta_t}\beta\tau+\frac{\beta_x}\beta\xi+\eta^1=a\lambda,\quad
\tau_t+\xi_x+\frac{\alpha_t}\alpha\tau+\frac{\alpha_x}\alpha\xi-\eta^1=c\lambda,\quad
\beta\eta^0=b\lambda,\quad
\eta^0_{tx}=d\alpha\lambda.
\end{gather}

If $a\ne0$, then we can make $a=1$ by rescaling the template-form equation~\eqref{eq:TemplateFormForCases5-10},
and thus $c\ne-1$.
The first two equations of the system~\eqref{eq:SplitClassifyingEqForCases5-10} are combined to
\begin{gather}\nonumber
\lambda=\frac{\beta_t}{\beta}\tau+\frac{\beta_x}{\beta}\xi+\eta^1,
\\\label{eq:ReducedClassifyingEqForCases5-10}
\tau_t+\left(\frac{\alpha_t}{\alpha}-c\frac{\beta_t}{\beta}\right)\tau+\xi_x+\left(\frac{\alpha_x}{\alpha}-c\frac{\beta_x}{\beta}\right)\xi=(c+1)\eta^1,
\end{gather}
and only the last equation plays the role of a classifying condition.
The third and fourth equations of the system~\eqref{eq:SplitClassifyingEqForCases5-10} merely establish a relation
between the constant parameters~$b$ and~$d$.
Indeed, in each of Cases~\ref{EKGcase5}--\ref{EKGcase9}, we have $(1/\beta)_{tx}$ is proportional to~$\alpha$,
$(1/\beta)_{tx}=C\alpha$ for some constant~$C$.
Hence we can always set $b=0$ by the equivalence transformation $\mathscr Z(b/\beta)$, 
which adds $bC$ to $\hat f$ and also makes $\eta^0=0$ and $d=0$.
In the general case $a\ne0$, the relation between~$b$ and~$d$ is $ad=bcC$.
For specific values of the parameter functions~$\alpha$ and~$\beta$ in \mbox{Cases~\ref{EKGcase5}--\ref{EKGcase9}},
we obtain the following values of the constant parameter~$C$ and the following forms of
the reduced classifying equation~\eqref{eq:ReducedClassifyingEqForCases5-10}:
\begin{gather*}
\mbox{Case~\ref{EKGcase5}}\colon\quad C=0,     \quad \tau_t+c\tau+\xi_x=(c+1)\eta^1; \\
\mbox{Case~\ref{EKGcase6}}\colon\quad C=1,     \quad \tau_t+c\tau+\xi_x+c\xi=(c+1)\eta^1;\\
\mbox{Case~\ref{EKGcase7}}\colon\quad C=q(q+1),\quad (x-t)(\tau_t+\xi_x)+(q+2+cq)(\tau-\xi)=(c+1)(x-t)\eta^1;\\
\mbox{Case~\ref{EKGcase8}}\colon\quad C=0,     \quad \tau_t+\xi_x-(q+2+cq)x^{-1}\xi=(c+1)\eta^1;\\
\mbox{Case~\ref{EKGcase9}}\colon\quad C=0,     \quad \tau_t+\xi_x=(c+1)\eta^1.
\end{gather*}
In each of Cases~\ref{EKGcase5}, \ref{EKGcase6}, \ref{EKGcase8} and~\ref{EKGcase9},
the corresponding classifying equation implies that there is a further Lie-symmetry extension if and only if
the associated value of the parameter function~$\hat f$ satisfies the equation~\eqref{eq:TemplateFormForCases5-10}
with arbitrary~$a\ne0$, $c$ and $b$ and with $d=bcC/a$.
Here the extension is given by Case~\ref{EKGcase12} with an arbitrary nonzero constant~$p$.
Case~\ref{EKGcase7} is similar
but the constant~$c$ is related to~$a$ according to $c=-(1+2/q)a$, $d=-(q+1)(q+2)b$,
and the extension for a fixed~$q$ is given by Case~\ref{EKGcase12}$_p$ with $p=2/q$.

If $a=0$, then $bc\ne0$, and we rescale the template-form equation~\eqref{eq:TemplateFormForCases5-10}
for making $c=1$.
Then
\begin{gather}\label{eq:ReducedClassifyingEqForCases5-10:a=0}
\frac{\beta_t}{\beta}\tau+\frac{\beta_x}{\beta}\xi+\eta^1=0,\quad
\lambda=\tau_t+\xi_x+\frac{\alpha_t}\alpha\tau+\frac{\alpha_x}\alpha\xi-\eta^1,\quad
\eta^0=b\frac\lambda\beta,\quad
\left(\frac\lambda\beta\right)_{tx}=\frac db\alpha\lambda.
\end{gather}
In Cases~\ref{EKGcase5}--\ref{EKGcase9},
the first two equations of~\eqref{eq:ReducedClassifyingEqForCases5-10:a=0} respectively reduce to
\begin{gather*}
\mbox{Case~\ref{EKGcase5}}\colon\ \tau=\eta^1,\ \ \lambda=\xi_x;\\ 
\mbox{Case~\ref{EKGcase6}}\colon\ \tau+\xi-\eta^1=0,\ \ \lambda=\tau_t+\xi_x;\quad 
\mbox{Case~\ref{EKGcase7}}\colon\ q\frac{\tau-\xi}{x-t}-\eta^1=0,\ \ \lambda=\tau_t+\xi_x+2\frac{\tau-\xi}{x-t};\\ 
\mbox{Case~\ref{EKGcase8}}\colon\ q\xi+\eta^1x=0,\ \ \lambda=\tau_t-2\frac{\eta^1}q;\quad 
\mbox{Case~\ref{EKGcase9}}\colon\ \eta^1=0,\ \ \lambda=\tau_t+\xi_x. 
\end{gather*}

In Case~\ref{EKGcase5}, we have no additional constraints on the parameters $b$, $c$ and~$d$.
The last two equations of~\eqref{eq:ReducedClassifyingEqForCases5-10:a=0} in view of the first ones just define $\eta^0$ and~$\xi$,
$\eta^0=b\xi_x{\rm e}^t$, $b\xi_{xx}=d\xi_x$.

In Cases~\ref{EKGcase6} and~\ref{EKGcase7}, the above equations imply $\lambda=0$, which contradicts the inequality $\lambda\ne0$.
This means that these cases possess no further Lie-symmetry extensions with $a=0$.

The system~\eqref{eq:ReducedClassifyingEqForCases5-10:a=0} implies $\tau_{tt}=0$ and $d=0$ in Case~\ref{EKGcase8}
and $d=0$ in Case~\ref{EKGcase9},
which correspond to the Lie-symmetry extensions to Cases~\ref{EKGcase11} and~\ref{EKGcase13}, respectively.

Merging the conditions derived separately for $a\ne0$ and~$a=0$,
we obtain the following proposition.

\begin{proposition}\label{pro:FurtherLieSymExtestionsForCases5-9}
A generalized nonlinear Klein--Gordon equation from Cases~\ref{EKGcase5}--\ref{EKGcase9} admits an additional Lie-symmetry extension
if and only if the corresponding value of the parameter function~$\hat f$ satisfies
an equation~\eqref{eq:TemplateFormForCases5-10} with $(a,b)\ne(0,0)$, where
$ad=0$ in Case~\ref{EKGcase5},
$ad=bc$ in Case~\ref{EKGcase6},
$c=-(1+2/q)a$, $d=-(q+1)(q+2)b$ in Case~\ref{EKGcase7}$_q$, and
$d=0$ in Cases~\ref{EKGcase8} and~\ref{EKGcase9}.
\end{proposition}

\begin{remark}\label{rem:NonlinKGEqsNormalizationOfSubclassesKn}
In view of the infinitesimal counterpart of~\cite[Proposition~10]{popo2010a},
all the subclasses~$\mathcal K'_N$, $N\in\Gamma$, of~$\mathcal K$ that are associated with strong Cases~\ref{EKGcase0}--\ref{EKGcase13}
are normalized, see Remark~\ref{rem:NonlinKGEqsOnMaxLieSymExtensions} for notation and definitions.
At the same time, this is not the case for most of the subclasses~$\mathcal K_N$, $N\in\Gamma$.
More specifically, the subclasses~$\mathcal K_1$, \dots, $\mathcal K_6$, $\mathcal K_{7_q}$, $\mathcal K_{8_q}$ and $\mathcal K_9$
are not normalized in view of the following arguments:
\begin{itemize}\itemsep=1ex
\item
$\mathcal K_1\supset\mathcal K_9$ \ but \ $G^\sim_1\nsupseteq G^\sim_9$ \
since \ $\mathscr I^0\in G^\sim_9\setminus G^\sim_1$;
\item
$\mathcal K_2\supset\mathcal K_9$ \ but \ $G^\sim_2\nsupseteq G^\sim_9$ \
since \ $\mathscr I^t\in G^\sim_9\setminus G^\sim_2$;
\item
$\mathcal K_3\ni K_f$ with $f={\rm e}^{-t}u^2$ \ but \ $\varpi_*G^\sim_3\nsupseteq G_f$ \
since \ $\varpi_*\big(\mathscr D^t(-\ln t)\circ\mathscr I^0\circ\mathscr D^t({\rm e}^{-t})\big)\in G_f\setminus\varpi_*G^\sim_3$;
\item
$\mathcal K_4\ni K_f$ \ with \ $f={\rm e}^{-p(t+x)}|u|^pu$ \ for any \ $p\ne0$ \ but \ $\varpi_*\mathfrak g^\sim_4\nsupseteq\mathfrak g_f$ \
since \ $D^t({\rm e}^{pt})\in\mathfrak g_f\setminus\varpi_*\mathfrak g^\sim_4$;
\item
$\mathcal K_5\ni K_f$ with $f={\rm e}^{-t}u^2$ \ but \ $\varpi_*G^\sim_5\nsupseteq G_f$ \
since \ $\varpi_*\big(\mathscr D^t(-\ln t)\circ\mathscr I^0\circ\mathscr D^t({\rm e}^{-t})\big)\in G_f\setminus\varpi_*G^\sim_5$;
\item
$\mathcal K_6\ni K_f$ \ with \ $f={\rm e}^{-p(t+x)}|u|^pu$ \ for any \ $p\ne0$ \ but \ $\varpi_*\mathfrak g^\sim_6\nsupseteq\mathfrak g_f$ \
since \ $D^t({\rm e}^{pt})\in\mathfrak g_f\setminus\varpi_*\mathfrak g^\sim_6$;
\item
$\mathcal K_{7,q}\supset\mathcal K_{12,p}$ with $p=2/q$ \ but \ $G^\sim_{7,q}\nsupseteq G^\sim_{12,p}$ \
since \ $\mathscr D^t(2t)\circ\mathscr D^x(x/2)\in G^\sim_{12,p}\setminus G^\sim_{7,q}$;
\item
$\mathcal K_{8,q}\supset\mathcal K_{12,p}$ with $p=2/q$ \ but \ $G^\sim_{8,q}\nsupseteq G^\sim_{12,p}$ \
since \ $\mathscr I^0\in G^\sim_{12,p}\setminus G^\sim_{8,q}$;
\item
$\mathcal K_9\supset\mathcal K_{13}=\{K_{{\rm e}^u}\}$ \ but \ $G^\sim_9\nsupseteq G^\sim_{13}$ \
since \ $\mathscr D^t(-t^{-1})\circ\mathscr Z(2\ln|t|)\in G^\sim_{13}\setminus G^\sim_9$.
\end{itemize}
Here $G_f$ denotes the point symmetry group of the equation~$K_f$.
The subclasses~$\mathcal K_0$, $\mathcal K_{10}$, $\mathcal K_{11}$, $\mathcal K_{12,p}$ and $\mathcal K_{13}$ are normalized
since they coincide with~$\mathcal K$, $\mathcal K'_{10}$ and the singletons $\{K_{{\rm e}^{u/x}}\}$, $\{K_{|u|^pu}\}$ and $\{K_{{\rm e}^u}\}$, respectively.%
\footnote{%
It is obvious that a class consisting of a single system of differential equations is normalized.
}
\end{remark}

\section{On group classification of subclasses}\label{sec:OnGroupClassificationOfSubclasses}

Although we have exhaustively solved the group classification problem for the class~$\mathcal K$,
which consists of the equations of the form~\eqref{eq:NonlinKGEqs},
this does not directly lead to the solution of the group classification problem for each of the subclasses of~$\mathcal K$.
Given a subclass~$\hat{\mathcal K}$ of~$\mathcal K$, Theorem~\ref{thm:NonlinKGEqsGroupClassification} is used
for the group classification of the subclass~$\hat{\mathcal K}$ with respect to its equivalence group~$G^\sim_{\hat{\mathcal K}}$ as follows.
\begin{itemize}
\item
Recalling the normalization of~$\mathcal K$,
construct the equivalence group~$G^\sim_{\hat{\mathcal K}}$ as the subgroup of~$G^\sim$
that consists of the elements of~$G^\sim$ preserving the subclass~$\hat{\mathcal K}$.
\item
For each $N\in\Gamma$, intersect the subclass~$\hat{\mathcal K}$ with the $G^\sim$-orbit $G^\sim_*\mathcal K'_N$ the subclass~$\mathcal K'_N$
(resp.\ with the $G^\sim$-orbit $G^\sim_*\mathcal K_N$ the subclass~$\mathcal K_N$).
This is realized via selecting those values of the arbitrary element~$f$ for equations from the orbit
that satisfy the additional auxiliary constraint singling out the subclass~$\hat{\mathcal K}$ from the class~$\mathcal K$.
The collection of the intersections presents a complete list of Lie symmetry extensions within the subclass~$\hat{\mathcal K}$.
\item
In the selected values of~$f$, gauge parameters by transformations from~$G^\sim_{\hat{\mathcal K}}$.
\end{itemize}
Since the subclass~$\hat{\mathcal K}$ is in general not normalized,
the above procedure looks easier than directly solving the group classification problem for the subclass~$\hat{\mathcal K}$
although its computational complexity is quite high.

The subclasses~$\mathcal K_N$, $N\in\Gamma$, of~$\mathcal K$
that are associated with weak Cases~\ref{EKGcase1}--\ref{EKGcase9},
cf.\ Remark~\ref{rem:NonlinKGEqsOnMaxLieSymExtensions},
are specific in regard to the above procedure.
The group classification of each of these subclasses up to the equivalence generated by the corresponding equivalence groupoid
can be easily derived via analyzing the Hasse diagram in Figure~\ref{fig:HasseDiagram},
which depicts the structure of the partially ordered set of these cases.
Nevertheless, this is not the case for the group classification up to the equivalence generated by the corresponding equivalence group
since most of the subclasses~$\mathcal K_N$, $N\in\Gamma$, are not normalized,
see Remark~\ref{rem:NonlinKGEqsNormalizationOfSubclassesKn}.
Note that both the group classifications of the subclasses~$\mathcal K_{10}$, $\mathcal K_{11}$, $\mathcal K_{12,p}$ and $\mathcal K_{13}$
are trivial since the kernel Lie invariance algebra of equations from each of these subclasses
is the maximal Lie invariance algebra for every such equation.

Consider in detail the subclass~$\mathcal K_2$ of the class~\eqref{eq:NonlinKGEqs},
which is related to weak Case~\ref{EKGcase2} of Theorem~\ref{thm:NonlinKGEqsGroupClassification}
and thus consists of the equations of the form
\begin{gather*}
u_{tx}=f(\omega,u), \quad\mbox{where}\quad \omega:=x-t,\quad f_{uu}\ne0,
\end{gather*}
or, in the variables $(\check t,\check x,\check u)=(x+t,x-t,u)$ with $\check f(\check x,\check u)=f(\omega,u)$,
\begin{gather*}
\check u_{\check t\check t}-\check u_{\check x\check x}=\check f(\check x,\check u),\quad \check f_{\check u\check u}\ne0.
\end{gather*}

\begin{lemma}\label{lem:NonlinKGEqsGsim2}
The equivalence group~$G^\sim_2$ of the class~$\mathcal K_2$ is constituted by the transformations of the form
\begin{gather}\label{eq:NonlinKGEqsGsim2}
\tilde t=c_1t+c_2,\quad
\tilde x=c_1x+c_3,\quad
\tilde u=c_4u+U^0(\omega),\quad
\tilde f=c_1^{-2}(c_4f-U^0_{\omega\omega})
\end{gather}
and the discrete equivalence transformation $\mathscr I^0$: $\tilde t= x$, $\tilde x=t$, $\tilde u=u$, $\tilde f=f$.
Here $c_1$, \dots,~$c_4$ are arbitrary constants with $c_1c_4\ne0$,
and $U^0$ is an arbitrary smooth function of $\omega:=x-t$.
\end{lemma}

\begin{proof}
Since the class~\eqref{eq:NonlinKGEqs} is normalized, the equivalence group of any subclass of~\eqref{eq:NonlinKGEqs}
is the subgroup of~$G^\sim$ that consists of the elements of~$G^\sim$ preserving this subclass.
It is obvious that the transformation~$\mathscr I^0$ belongs to~$G^\sim_2$.
Any transformation of the form~\eqref{eq:NonlinKGEqsGsim0}
that is contained by the group~$G^\sim_2$ satisfies the equation
\begin{gather}\label{eq:NonlinKGEqsGsim2Condition}
T_tX_x\tilde f(X-T,U)=Cf(x-t,u)+U^0_{tx}.
\end{gather}
We act on~\eqref{eq:NonlinKGEqsGsim2Condition} by the operator $\p_t+\p_x$, obtaining
\begin{gather*}
(X_x-T_t)\,\tilde f_{\tilde\omega}(X-T,U)+(U^0_t+U^0_x)\,\tilde f_{\tilde u}(X-T,U)+\frac{T_{tt}X_x+T_tX_{xx}}{T_tX_x}\,\tilde f(X-T,U)\\
\qquad=\frac{U^0_{ttx}+U^0_{txx}}{T_tX_x}.
\end{gather*}
Since $\tilde f$ is an unconstrained value of the arbitrary element of the class~$\mathcal K_2$,
for computing~$G^\sim_2$, we can split the last equation with respect to $\tilde f$ and its derivatives.
As a result, we derive the equations $T_t=X_x$ and $U^0_t+U^0_x=0$ on the parameters involved in the form~\eqref{eq:NonlinKGEqsGsim0}.
The integration of these equations implies the form~\eqref{eq:NonlinKGEqsGsim2}.
\end{proof}

\looseness=-1
The appropriate subalgebra $\mathfrak s=\langle\p_t+\p_x\rangle$ of $\mathfrak g_\spanindex$
is the kernel Lie invariance algebra of the equations from the class~$\mathcal K_2$.
In other words, Case~\ref{EKGcase2} is the general case with no Lie-symmetry extensions
within this class.
It is not normalized since
the action groupoid of~$G^\sim_2$ is properly contained in the equivalence groupoid~$\mathcal G^\sim_2$.
Indeed, many equations in~$\mathcal K_2$, e.g., the Liouville equation,
possess points symmetries that are not related to equivalence transformations of~$\mathcal K_2$. 
Moreover, Theorem~\ref{thm:NonlinKGEqsSuclassOfCase2AdmTrans} below implies 
that the class~$\mathcal K_2$ is not semi-normalized as well. 
This is why it is natural that the group classifications of the class~$\mathcal K_2$
up to the $\mathcal G^\sim_2$- and the \mbox{$G^\sim_2$-equivalences} are different.
We easily see from the Hasse diagram in Figure~\ref{fig:HasseDiagram} that
\mbox{$\mathcal G^\sim_2$-inequivalent} cases of Lie-symmetry extensions within the class~$\mathcal K_2$
are exhausted by Cases~\ref{EKGcase6}, \ref{EKGcase7}, \ref{EKGcase9}, \ref{EKGcase10}, \ref{EKGcase12} and~\ref{EKGcase13},
which gives the complete group classification of this class up to the $\mathcal G^\sim_2$-equivalence.
(In~Case~\ref{EKGcase6}, we should additionally alternate the sign of~$x$.)

The complete group classification of the class~$\mathcal K_2$ up to the $G^\sim_2$-equivalence is more delicate.
It can be derived from the group classification of the superclass~$\mathcal W$ of~$\mathcal K_2$,
which consists of the equations of the following form in the variables $(\check t,\check x,\check u)$:
\begin{gather*}
\check u_{\check t\check t}-\check g(\check x,\check u)\check u_{\check x\check x}=\check f(\check x,\check u),\quad
(\check g_{\check u},\check f_{\check u\check u})\ne(0,0).
\end{gather*}
It is obvious that the class~$\mathcal K_2$ is singled out from the superclass~$\mathcal W$ by the constraint $\check g=1$.
The comprehensive group analysis of the class~$\mathcal W$ was carried out in~\cite{vane2020b},
where a different notation of the arbitrary elements~$\check g$ and~$\check f$ was used,
$\check g\rightsquigarrow f$ and $\check f\rightsquigarrow g$.
The equivalence group~$G^\sim_{\mathcal W}$ and the equivalence groupoid~$\mathcal G^\sim_{\mathcal W}$ of~$\mathcal W$
were described in~\cite[Theorem~6]{vane2020b} and in~\cite[Theorem~9]{vane2020b}, respectively.
The action groupoid of~$G^\sim_{\mathcal W}$ is a proper subgroupoid of the groupoid~$\mathcal G^\sim_{\mathcal W}$, 
i.e., the superclass~$\mathcal W$ is not normalized.
The restriction of the action groupoid of~$G^\sim_{\mathcal W}$ to the subclass~$\mathcal K_2$ of~$\mathcal W$
coincides with the action groupoid of~$G^\sim_2$.
This is why the complete group classification of the class~$\mathcal K_2$ up to the $G^\sim_2$-equivalence
can be singled out from the complete group classification of the superclass~$\mathcal W$ up to the $G^\sim_{\mathcal W}$-equivalence,
which was presented in~\cite[Theorem~8]{vane2020b}.
Since the gauge $\check g=1$ modulo the $G^\sim_{\mathcal W}$-equivalence
was used for representatives of Lie-symmetry extensions whenever it was possible,
to classify Lie symmetries of equations from~$\mathcal K_2$ up to the $G^\sim_2$-equivalence
it suffices to select all the cases of~\cite[Table~1]{vane2020b} with $\check g=1$,
i.e., $f=1$ in the notation of~\cite{vane2020b}, write them in the variables $(t,x,u)$,
and supplement the result with Cases~\ref{EKGcase6} and~\ref{EKGcase7}
of Theorem~\ref{thm:NonlinKGEqsGroupClassification} of the present paper,
which are the counterparts of the appropriate portions of Cases~1 and~2 of~\cite[Table~1]{vane2020b}.
As a result, we prove the following theorem.

\begin{theorem}\label{thm:NonlinKGEqsSuclassOfCase2GroupClassification}
A complete list of $G^\sim_2$-inequivalent cases of Lie-symmetry extensions
of the kernel Lie invariance algebra~$\mathfrak g^\cap=\langle\p_t+\p_x\rangle$ in the class~$\mathcal K_2$
are exhausted by the following cases:
\begin{align*}
2.\ &
\mbox{General case } f=\hat f(x-t,u)\colon \ \mathfrak g_f=\langle\p_t+\p_x\rangle;\\[.5ex]
6.\ &
f={\rm e}^{-x+t}\hat f({\rm e}^{x-t}u)\colon \ \mathfrak g_f=\langle\p_t+u\p_u,\,\p_x-u\p_u\rangle;\\[.5ex]
7.\ &
f=|x-t|^{-q-2}\hat f(|x-t|^qu), \ q\ne0\colon \ \mathfrak g_f=\langle\p_t+\p_x,\,t\p_t+x\p_x-qu\p_u\rangle;\\[.5ex]
9{\rm a}.\ &
f=\hat f(u)\colon \ \mathfrak g_f=\langle\p_t,\,\p_x,\,t\p_t-x\p_x\rangle;\\[.5ex]
9{\rm b}.\ &
f=\hat f(u){\rm e}^{x-t}\colon \ \mathfrak g_f=\langle{\rm e}^t\p_t,\,{\rm e}^{-x}\p_x,\,\p_t+\p_x\rangle;\\[.5ex]
10{\rm a}.\ &
f=\hat f(u)(x-t)^{-2}\colon \ \mathfrak g_f=\langle\p_t+\p_x,\,t\p_t+x\p_x,\,t^2\p_t+x^2\p_x\rangle;\\[.5ex]
10{\rm b}.\ &
f=\hat f(u)\cos^{-2}(x-t)\colon \ \mathfrak g_f=\langle\p_t+\p_x,\,\cos2t\,\p_t-\cos2x\,\p_x,\,\sin2t\,\p_t-\sin2x\,\p_x\rangle;\\[.5ex]
10{\rm c}.\ &
f=\hat f(u)\cosh^{-2}(x-t)\colon \ \mathfrak g_f=\langle\p_t+\p_x,\,{\rm e}^{2t}\p_t-{\rm e}^{2x}\p_x,\,{\rm e}^{-2t}\p_t-{\rm e}^{-2x}\p_x\rangle;\\[.5ex]
10{\rm d}.\ &
f=\hat f(u)\sinh^{-2}(x-t)\colon \ \mathfrak g_f=\langle\p_t+\p_x,\,{\rm e}^{2t}\p_t+{\rm e}^{2x}\p_x,\,{\rm e}^{-2t}\p_t+{\rm e}^{-2x}\p_x\rangle;\\[.5ex]
12{\rm a}.\ &
f=|u|^p u, \ p\ne-1,0\colon \ \mathfrak g_f=\langle\p_t,\,\p_x,\,t\p_t-x\p_x,\,-pt\p_t+u\p_u\rangle;\\[.5ex]
12{\rm b}.\ &
f=|u|^p u{\rm e}^{x-t}, \ p\ne-1,0\colon \ \mathfrak g_f=\langle{\rm e}^t\p_t,\,{\rm e}^{-x}\p_x,\,\p_t+\p_x,\,p\p_t+u\p_u\rangle;\\[.5ex]
13.\ &
f={\rm e}^u\colon \ \mathfrak g_f=\langle\tau(t)\p_t+\xi(x)\p_x-(\tau_t(t)+\xi_x(x))\p_u\rangle.
\end{align*}
Here $\hat f$ is an arbitrary smooth function of its arguments
whose second derivative with respect to the argument involving~$u$ is nonzero,
$q$ and~$p$ are arbitrary constants that satisfy the conditions indicated in the corresponding cases.
In Case~\ref{EKGcase13}, the components~$\tau$ and~$\xi$
run through the sets of smooth functions of~$t$ or~$x$, respectively.
\end{theorem}

We use the two-level numeration for the classification cases listed in Theorem~\ref{thm:NonlinKGEqsSuclassOfCase2GroupClassification}
for indicating the presence of additional equivalences between these cases.
Namely, numbers with the same Arabic numerals and different Roman letters
correspond to cases that are $G^\sim_2$-inequivalent
but $\mathcal G^\sim_2$-equivalent and hence $G^\sim$-equivalent
as Lie-symmetry extensions within the class~$\mathcal K$.
Related cases in Theorems~\ref{thm:NonlinKGEqsGroupClassification} and~\ref{thm:NonlinKGEqsSuclassOfCase2GroupClassification}
have numbers with the same Arabic numerals.

To find all additional equivalence transformations among $G^\sim_2$-inequivalent classification cases for~$\mathcal K_2$
and thus to relate the group classification of~$\mathcal K_2$ modulo the $G^\sim_2$-equivalence
to that modulo the $\mathcal G^\sim_2$-equivalence,
we need to classify admissible transformations within the class~$\mathcal K_2$ up to the $G^\sim_2$-equivalence.
The description of the equivalence groupoid~$\mathcal G^\sim_2$ of the class~$\mathcal K_2$
can be derived from~\cite[Theorem~9]{vane2020b}
analogously to the above derivation of the group classification of~$\mathcal K_2$ up to the $G^\sim_2$-equivalence.
See necessary notions in~\cite[Section~2]{vane2020b}.

\begin{theorem}\label{thm:NonlinKGEqsSuclassOfCase2AdmTrans}
A generating (up to the $G^\sim_2$-equivalence) set of admissible transformations for the class~$\mathcal K_2$,
which is minimal and self-consistent with respect to the $G^\sim_2$-equivalence,
is the union of the following families of admissible transformations $(f,\Phi,\tilde f)$:
\begin{align*}
{\rm T1 }.\ & f= \hat f(u),\ \ \tilde f=-f,\quad \Phi\colon\ \ \tilde t=-t,\ \ \tilde x=x,\ \ \tilde u=u;\\[.5ex]
{\rm T2 }.\ & f= \hat f(u),\ \ \tilde f=f,\quad \Phi\colon\ \ \tilde t=t{\rm e}^\gamma,\ \ \tilde x=x{\rm e}^{-\gamma},\ \ \tilde u=u,\ \ \gamma\in\mathbb R_{\ne0};\\[.5ex]
{\rm T3 }.\ & f=\hat f(u){\rm e}^{x-t},\ \ \tilde f=\hat f(\tilde u),\quad \Phi\colon\ \ \tilde t=-{\rm e}^{-t},\ \ \tilde x={\rm e}^x,\ \ \tilde u=u;\\[.5ex]
{\rm T4a}.\ & f= \hat f(u)(x-t)^{-2},\ \ \tilde f=\hat f(\tilde u)(\tilde x-\tilde t)^{-2},\quad \Phi\colon\ \ \tilde t=t^{-1},\ \ \tilde x=x^{-1},\ \ \tilde u=u;\\[.5ex]
{\rm T4b}.\ & f=-\hat f(u)\cos^{-2}(x-t),\ \ \tilde f=\hat f(\tilde u)(\tilde x-\tilde t)^{-2},\quad \Phi\colon\ \ \tilde t=\tan t,\ \ \tilde x=\cot x,\ \ \tilde u=u;\\[.5ex]
{\rm T4c}.\ & f=-\hat f(u)\cosh^{-2}(x-t),\ \ \tilde f=\hat f(\tilde u)(\tilde x-\tilde t)^{-2},\quad \Phi\colon\ \ \tilde t=-\tfrac12{\rm e}^{2t},\ \ \tilde x=\tfrac12{\rm e}^{2x},\ \ \tilde u=u;\\[.5ex]
{\rm T4d}.\ & f= \hat f(u)\sinh^{-2}(x-t),\ \ \tilde f=\hat f(\tilde u)(\tilde x-\tilde t)^{-2},\quad \Phi\colon\ \ \tilde t=\tfrac12{\rm e}^{2t},\ \ \tilde x=\tfrac12{\rm e}^{2x},\ \ \tilde u=u;\\[.5ex]
{\rm T5 }.\ & f={\rm e}^u,\ \ \tilde f={\rm e}^{\tilde u},\quad \Phi\colon\ \ \tilde t=T(t),\ \ \tilde x=X(x),\ \ \tilde u=u-\ln(T_tX_x),
\end{align*}
where
$(T,X)$ runs through a complete set of representatives of cosets of $(T,X)$ with $T_tX_x>0$ and $(T_{tt},X_{xx})\ne(0,0)$
with respect to the action of the group constituted by the transformations of the form
$\hat t=c_1t+c_2$, $\hat x=c_1x+c_3$, $\hat T=\tilde c_1T+\tilde c_2$, $\hat X=\tilde c_1X+\tilde c_3$,
where $c_1$, $c_2$, $c_3$, $\tilde c_1$, $\tilde c_2$ and~$\tilde c_3$ are arbitrary constants with $c_1\tilde c_1\ne0$.
\end{theorem}

Having the classification of admissible transformations within the class~$\mathcal K_2$ up to the $G^\sim_2$-equivalence,
we can directly find all independent additional equivalence transformations among classification cases
listed in Theorem~\ref{thm:NonlinKGEqsSuclassOfCase2GroupClassification}.
These transformations are
\begin{gather*}
{\rm T1 }\colon\ \mbox{Case~9a}_{\hat f}\to\mbox{Case~9a}_{-\hat f},\ \
{\rm T3 }\colon\ \mbox{Case~9b} \to\mbox{Case~9a},\ \mbox{Case~12b} \to\mbox{Case~12a},\\ 
{\rm T4b}\colon\ \mbox{Case~10b}\to\mbox{Case~10a},\ \
{\rm T4c}\colon\ \mbox{Case~10c}\to\mbox{Case~10a},\ \
{\rm T4d}\colon\ \mbox{Case~10d}\to\mbox{Case~10a}.
\end{gather*}

\section{Conclusion}\label{sec:Conclusion}

In the present paper, 
we have carried out the complete (contact) group classification
of the class~\eqref{eq:NonlinKGEqs} of (1+1)-dimensional generalized nonlinear Klein--Gordon equations up to the $G^\sim$-equivalence.
This has substantially enhanced the results on Lie symmetries of such equations
that were obtained in the seminal papers~\cite{lahn2005a,lahn2006a}.
At first, extending results of Lie's paper~\cite{lie1881b}, we have shown in Lemma~\ref{lem:GenNonlinKGEqsContactAdmTrans} 
that any contact admissible transformation within the class~\eqref{eq:NonlinKGEqs} 
is the first-order prolongation of a point admissible transformation within this class.
In other words, the study of contract-transformation structures related to equations from the class~\eqref{eq:NonlinKGEqs} 
reduces to the study of their point-transformation counterparts. 
We have proved in Lemma~\ref{lem:NonlinKGEqsGsim} that the class~\eqref{eq:NonlinKGEqs} is normalized.
Therefore, applying the algebraic method, we have reduced the group classification of~\eqref{eq:NonlinKGEqs}
to classifying the appropriate subalgebras of the projection
$\varpi_*\mathfrak g^\sim=\mathfrak g_\spanindex$
of the equivalence algebra~$\mathfrak g^\sim$.
In addition to this, we have employed the specific structure of~$\mathfrak g^\sim$ for twofold involving
the classical Lie theorem on realizations of Lie algebras by vector fields on the line~\cite{lie1880a}
into the classification procedure.
Moreover, the normalization of the class~\eqref{eq:NonlinKGEqs} means
that the action groupoid~\cite{vane2020b} of the equivalence group~$G^\sim$ coincides
with the entire equivalence groupoid~$\mathcal G^\sim$
of the class~\eqref{eq:NonlinKGEqs}.
Hence the complete group classification of the class~\eqref{eq:NonlinKGEqs} up to the $G^\sim$-equivalence
coincides with its complete group classification up to the $\mathcal G^\sim$-equivalence,
which is just the general point equivalence within this class.
In other words, we have no additional point equivalences
between $G^\sim$-inequivalent classification cases.

Lie symmetries of equations from the class~\eqref{eq:NonlinKGEqs} were considered in Section~6 of~\cite{lahn2006a},
and cases with two-, three- and four-dimensional Lie invariance algebras
were listed in Table~1 therein, see also Section~V and Table~I in~\cite{lahn2005a}.
Cases~1--6, 8 and~9 of Table~1 and the equation~(5.4) in~\cite{lahn2006a} correspond to
Cases~\ref{EKGcase5}, \ref{EKGcase6}, \ref{EKGcase7}$_{q=1}$, \ref{EKGcase8}$_{q=-1}$,
\ref{EKGcase10}, \ref{EKGcase11}, \ref{EKGcase9}, \ref{EKGcase12} and~\ref{EKGcase13}
of Theorem~\ref{thm:NonlinKGEqsGroupClassification} in the present paper.
Case~7 of Table~1 in~\cite{lahn2006a} should be excluded from the classification since it is equivalent to Case~9 therein,
see the discussion of the case $(m,n,k)=(2,3,1)$ in Section~\ref{sec:NonlinKGEqsProof}.
Cases~\ref{EKGcase7}$_{q\ne1}$ and \ref{EKGcase8}$_{q\ne-1}$ were missed in~\cite{lahn2006a}
owing to superfluously constraining the parameter~$q$, see Remark~\ref{rem:NonlinKGEqsGaugingPInCases7And8}.

We have additionally enhanced the results of~\cite{lahn2005a,lahn2006a}
by explicitly singling out the equations from the class~\eqref{eq:NonlinKGEqs}
with infinite-dimensional maximal Lie invariance algebras in Lemma~\ref{lem:NonlinKGEqsConditionsForAppropriateSubalgebras}.
It turns out that any such equation is $G^\sim$-equivalent to the Liouville equation.
For the other equations from the class~\eqref{eq:NonlinKGEqs},
whose maximal Lie invariance algebras are finite-dimensional,
we have found the least upper bound of dimensions of these algebras,
which is equal to four.
One more tool for arranging the classification is to assign a value of the triple $(m,n,k)$
of~$G^\sim$-invariant integers to each case of Lie-symmetry extension in the class~\eqref{eq:NonlinKGEqs}.
We have strongly restricted the set of candidates for appropriate values of the triple
at the stage of preliminary analysis
using Lemma~\ref{lem:NonlinKGEqsConditionsForAppropriateSubalgebras} and the Lie theorem.
The final selection of the appropriate values has been done in the course of the group classification.
It was important for simplifying the computations on all classification stages
that for equations from the class~\eqref{eq:NonlinKGEqs}, in contrast to evolution equations, 
the Lie theorem can be applied to both the $t$- and the $x$-projections of Lie-symmetry vector fields.

Although the characteristic triple $(m,n,k)$ has a simple interpretation
and is principal for the proof of Theorem~\ref{thm:NonlinKGEqsGroupClassification},
it does not suffice for completely distinguishing $G^\sim$-inequivalent classification cases.
This is why we attempted to find as many~$G^\sim$-invariant integer characteristics of the classification cases as possible,
and have found even twelve of them in total, $m$, $n$, $k$, $l$, $j_1$, $j_2$, $j_{12}$, $j_{13}$, $j_{23}$, $r_1$, $r_2$ and~$r_3$.
The comprehensive analysis has shown that the complete tuple of these twelve characteristics is redundant.
As found out in Remark~\ref{rem:IdentifyingInvValues},
these characteristics can constitute no pairs and exactly eight triples
that suffice for distinguishing $G^\sim$-inequivalent classification cases,
and the most remarkable triple is $(r_3,j_1,r_2)$.
The same triple is, simultaneously with~$n$, of primary importance for identifying the pairs of $G^\sim$-inequivalent weak classification cases
that do not represent successive Lie-symmetry extensions,
i.e., the pairs $(\mbox{Case\ }N,\mbox{Case\ }\bar N)$ with $\mbox{Case\ }N\nprec\mbox{Case\ }\bar N$,
see Section~\ref{sec:SuccessiveLie-SymExtensions}.
To be sufficient for this task, the tuple $(n,r_3,r_2,j_1)$
should be extended with one of the characteristics~$k$, $m$ or~$r_1$,
and the thus obtained three tuples exhaust the set of such sufficient tuples of minimum size, which is equal to five.
In the same section, we have directly checked
that all the other pairs of $G^\sim$-inequivalent classification cases
are indeed associated with successive Lie-symmetry extensions.
The consideration has been summarized in the Hasse diagram in Figure~\ref{fig:HasseDiagram},
which represents the structure of the partially ordered set of $G^\sim$-inequivalent Lie-symmetry extensions
within the class~\eqref{eq:NonlinKGEqs}.
Analyzing the Hasse diagram allows one to easily solve the group classification problems
up to the general point equivalence
for the subclasses~$\mathcal K_N$, $N\in\Gamma$, of~$\mathcal K$,
which correspond, under the interpretation in the weak sense, to the classification cases
that have been listed in Theorem~\ref{thm:NonlinKGEqsGroupClassification}; see Remark~\ref{rem:NonlinKGEqsOnMaxLieSymExtensions}.
Since the subclasses~$\mathcal K_1$, \dots, $\mathcal K_6$, $\mathcal K_{7,q}$, $\mathcal K_{8,q}$ and $\mathcal K_9$ are not normalized,
the group classification of any such~$\mathcal K_N$ up to the $G^\sim_N$-equivalence is not so easy.
The last claim has been illustrated in Section~\ref{sec:OnGroupClassificationOfSubclasses}
by carrying out the group classification of the subclass~$\mathcal K_2$ up to the $G^\sim_2$-equivalence.
Therein, we have also discussed a procedure of using Theorem~\ref{thm:NonlinKGEqsGroupClassification}
for the group classification of any subclass of the class~$\mathcal K$ with respect to the equivalence group of this subclass.

The classification of Lie symmetries is the first necessary step
for extended symmetry analysis of equations from the class~\eqref{eq:NonlinKGEqs}.
It can be used for the classification of Lie reductions and further finding exact invariant solutions of these equations.
Since the general solution of the Liouville equation is well known,
Lie reductions should be carried out only for equations from the class~\eqref{eq:NonlinKGEqs}
with finite-dimensional maximal Lie invariance algebras.
In view of Lemma~\ref{lem:NonlinKGEqsConditionsForAppropriateSubalgebras}(iii),
which states upper bound four for the dimensions of such algebras,
the classification of subalgebras of three- and four-dimensional Lie algebras in~\cite{pate1977a}
is extremely relevant here.
As an example, consider Case~\ref{EKGcase10}.
It is the only case among those with finite-dimensional maximal Lie invariance algebras,
where the algebra~$\mathfrak g_f$ is not solvable.
More precisely, it is isomorphic to the algebra~$\mathfrak{sl}(2,\mathbb R)$,
and its inequivalent one-dimensional subalgebras and the associated Lie reductions to ordinary differential equations
are the following:
\begin{gather*}
1.\ \langle\p_t+\p_x\rangle\colon\ \ u=\varphi(\omega),\ \ \omega=x-t,\ \ \varphi_{\omega\omega}=-\hat f(\varphi)\omega^{-2};\\
2.\ \langle t\p_t+x\p_x\rangle\colon\ \ u=\varphi(\omega),\ \ \omega=\tfrac12\ln|x|-\tfrac12\ln|t|,\ \ \varphi_{\omega\omega}=-\hat f(\varphi)\sinh^{-2}\omega;\\
3.\ \langle(1+t^2)\p_t+(1+x^2)\p_x\rangle\colon\ \ u=\varphi(\omega),\ \ \omega=\arctan x-\arctan t,\ \ \varphi_{\omega\omega}=-\hat f(\varphi)\sin^{-2}\omega.
\end{gather*}

The knowledge of Lie symmetries of equations from the class~\eqref{eq:NonlinKGEqs} is also needed
for classification of reduction operators of these equations.
In the course of classifying reduction operators, it is natural to exclude those of them
that are induced by Lie symmetries, i.e., to look only for non-Lie reduction operators.
Unfortunately, the general description of regular reduction operators for equations from the class~$\mathcal K$ as a whole
from~\cite{yeho2010a} cannot be used in a reasonable way 
for describing reduction operators of particular equations from this class 
or for classifying reduction operators of equations constituting its proper subclasses.
Recall that regarding singular reduction operators of equations from the class~\eqref{eq:NonlinKGEqs},
a systematic study has been done in the literature
only for equations with $f=\hat f(u)$,
which constitute the subclass~$\mathcal K_9$ associated with 
Case~\ref{EKGcase9} of Theorem~\ref{thm:NonlinKGEqsGroupClassification},
see~\cite[Section~6]{kunz2008a}.

The consideration of the presented paper
can be extended to the much wider superclass of generalized nonlinear Klein--Gordon equations of the form
\begin{gather}\label{eq:NonlinKGEqsSuperclass}
u_{tx}=f(t,x,u,u_t,u_x).
\end{gather}
A preliminary study has shown
that Lemma~\ref{lem:GenNonlinKGEqsContactAdmTrans} may be generalized to this superclass, 
and the Lie theorem should be relevant for its group classification 
in the same way as for that of the class~\eqref{eq:NonlinKGEqs}.
The principal precondition for applying the Lie theorem 
to the group classification of the superclass~\eqref{eq:NonlinKGEqsSuperclass}
is to prove an analogue of Lemma~\ref{lem:NonlinKGEqsConditionsForAppropriateSubalgebras}
for this superclass, which seems a much more difficult problem
than proving Lemma~\ref{lem:NonlinKGEqsConditionsForAppropriateSubalgebras}.
In particular, one needs to single out, within the superclass~\eqref{eq:NonlinKGEqsSuperclass},
the equations with infinite-dimensional maximal Lie invariance algebras.

\section*{Acknowledgments}
The authors are sincerely grateful to the anonymous referees for a number of valuable remarks and suggestions. 
The authors thank Michael Kunzinger, Dmytro Popovych, Galyna Popovych and Olena Vaneeva for helpful discussions 
and acknowledge the partial financial support provided by the NAS of Ukraine under the project 0116U003059. 
The research of ROP was supported by the Austrian Science Fund (FWF), projects P25064 and P28770.


\end{document}